%% file: arxiv_version.tex
\documentclass[11pt, a4paper]{article}
\usepackage[margin=0.81in]{geometry}

\usepackage{microtype}
\usepackage{graphicx}
\usepackage{subfigure}
\usepackage{booktabs} 
\usepackage{amsmath}
\usepackage{amsthm}
\usepackage{amssymb}
\usepackage{wrapfig}
\usepackage{graphicx}
\usepackage[dvipsnames]{xcolor}  
\usepackage{multirow}
\usepackage{epsfig,amssymb,amsfonts,amsmath,amsthm} 
\usepackage{tikzit}
\input{basic.tikzstyles}
\usepackage{appendix}
\usepackage[ruled,vlined,linesnumbered]{algorithm2e}
\usepackage{hhline} 
\usepackage{hyperref}
\usepackage{times}
\usepackage{placeins}

\newcommand{\simplify}[1]{\sigma \circ \left(#1\right)}

\newcommand{\In}{\mathrm{in}}
\newcommand{\Out}{\mathrm{out}}
\newcommand{\ppr}{\mathrm{pr}}
\newcommand{\pr}{\mathrm{pr}}
\newcommand{\apr}{\mathrm{apr}}
\newcommand{\bipart}[3][]{\beta_{#1}(#2,#3)}
\newcommand{\cond}[2][]{\Phi_{#1}(#2)}
\newcommand{\union}{\cup}
\newcommand{\intersect}{\cap}
\newcommand{\sweepset}[2]{S_{#2}^{#1}}
\newcommand{\pjsweep}{\sweepset{p}{j}}
\newcommand{\bigo}[1]{O\!\left(#1\right)}

\newcommand{\p}{\mathbb{P}}

\newcommand{\abs}[1]{\left\lvert#1\right\rvert}
\newcommand{\cardinality}[1]{\abs{#1}}
\newcommand{\lscurve}[2][p]{#1\!\left[#2\right]}
\include{commands}

\usepackage{algorithmic}

\title{Local Algorithms for Finding Densely Connected Clusters}
\author{%
Peter Macgregor\thanks{University of Edinburgh, \texttt{peter.macgregor@ed.ac.uk}}
\and
He Sun\thanks{University of Edinburgh, \texttt{h.sun@ed.ac.uk}}}
\date{}

\begin{document}

\maketitle

\begin{abstract}
 Local graph clustering is an important algorithmic technique for analysing massive graphs, and has been widely applied in many research fields of data science.
While the objective of most (local) graph clustering algorithms is to find a vertex set of low conductance, there has been a sequence of recent studies that highlight the importance of the inter-connection between clusters
when analysing real-world datasets. Following this line of research, in this work we study local algorithms for finding a pair of vertex sets defined with respect to their inter-connection and their relationship with  the rest of the graph. The key to our analysis is a  new reduction technique that relates the structure of multiple sets to a \emph{single} vertex set in the reduced graph. Among many potential applications, we   show  that our algorithms  successfully recover  densely connected clusters in the Interstate Disputes Dataset and the US Migration Dataset.
\end{abstract}
 
\begin{figure*}[t]
\centering
    \subfigure[1816-1900]  {\includegraphics[width=0.24\columnwidth]{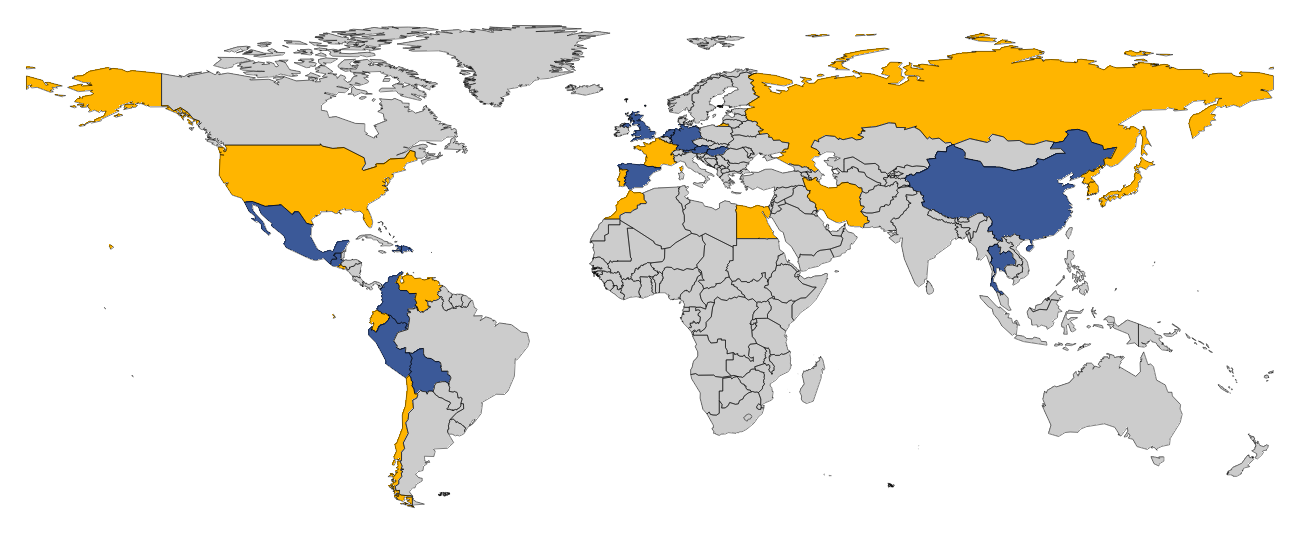}}
    \subfigure[1900-1950] {\includegraphics[width=0.24\columnwidth]{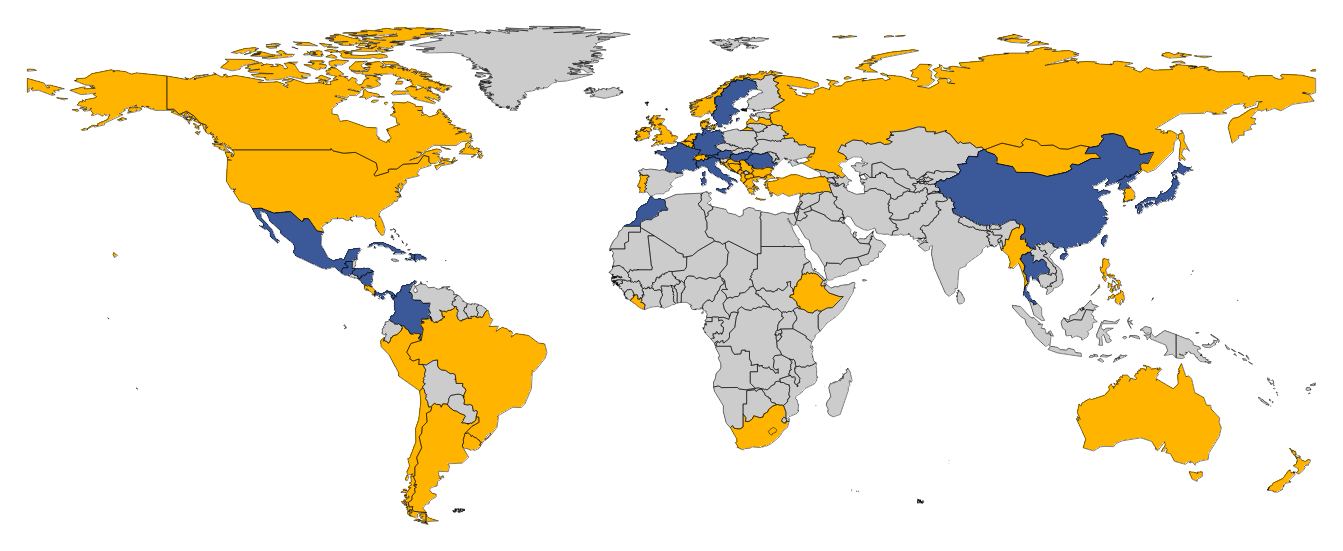}}
    \subfigure[1950-1990] {\includegraphics[width=0.24\columnwidth]{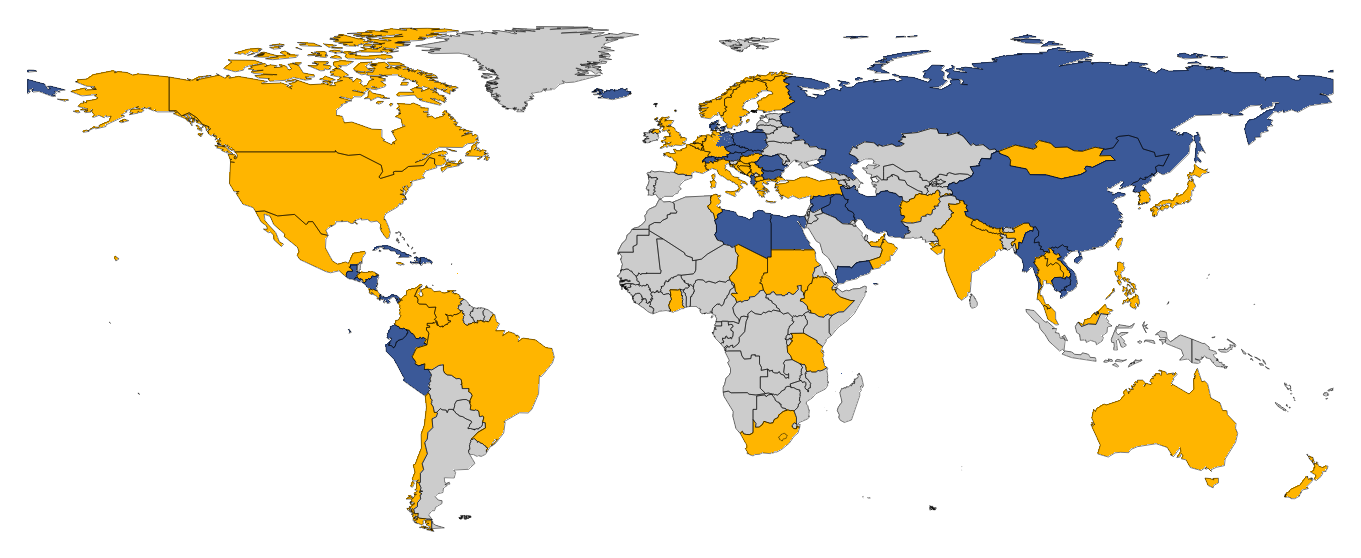}}
    \subfigure[1990-2010] {\includegraphics[width=0.24\columnwidth]{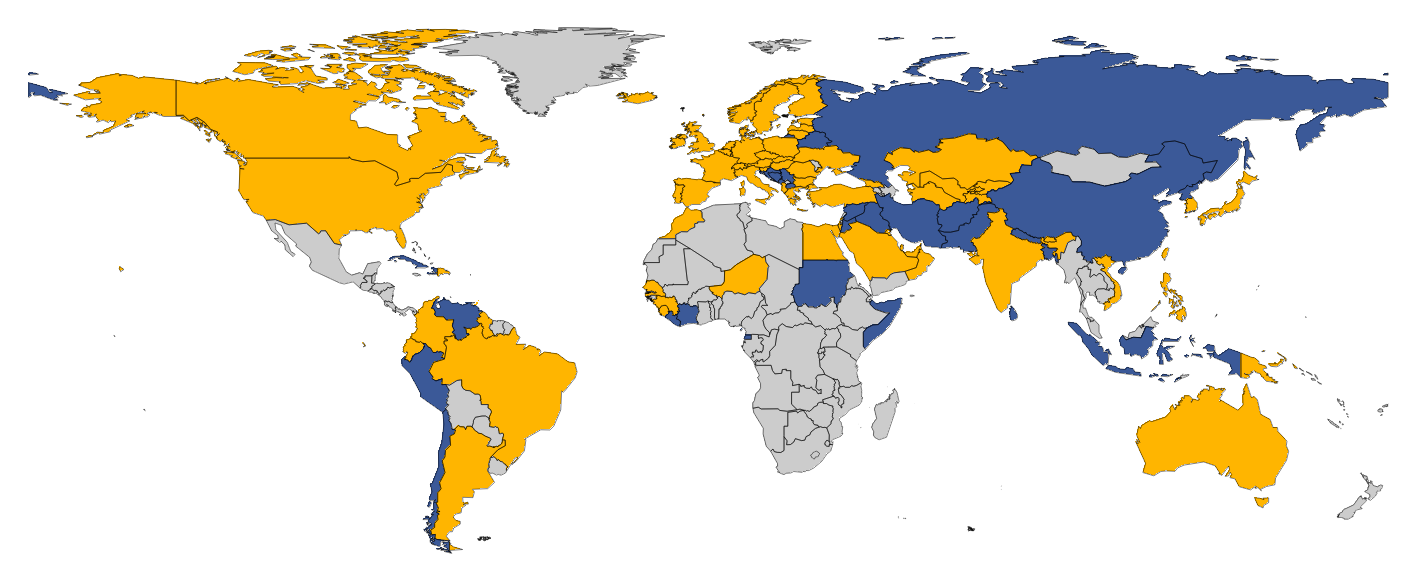}}
    \subfigure[Ohio Seed] {\includegraphics[width=0.24\columnwidth]{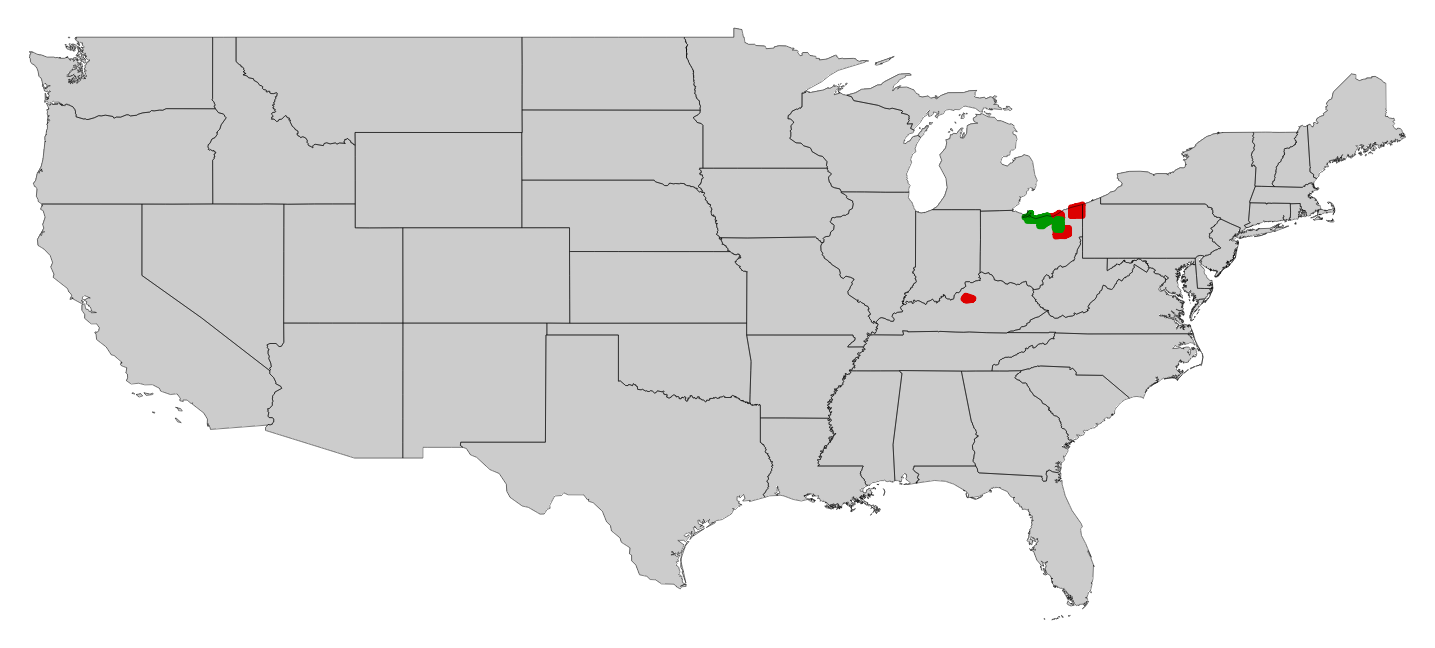}}
    \subfigure[New York Seed] {\includegraphics[width=0.24\columnwidth]{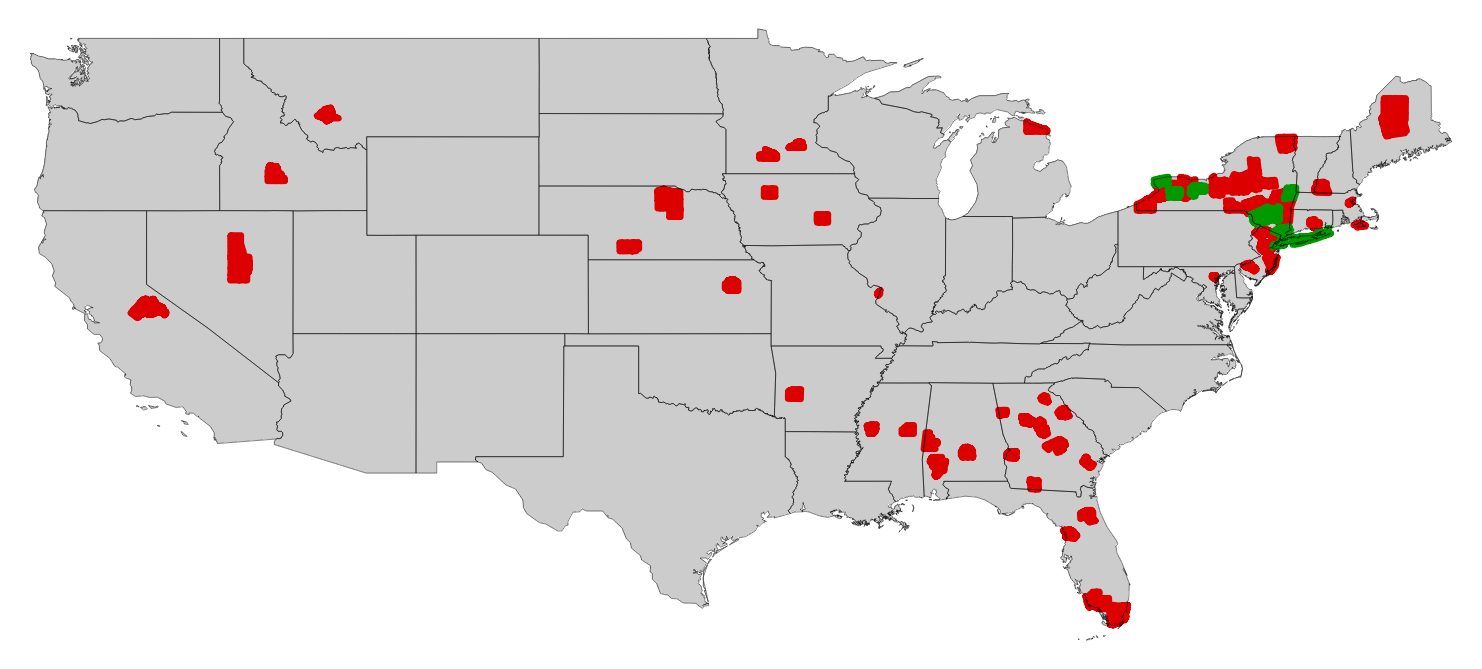}}
    \subfigure[California Seed] {\includegraphics[width=0.24\columnwidth]{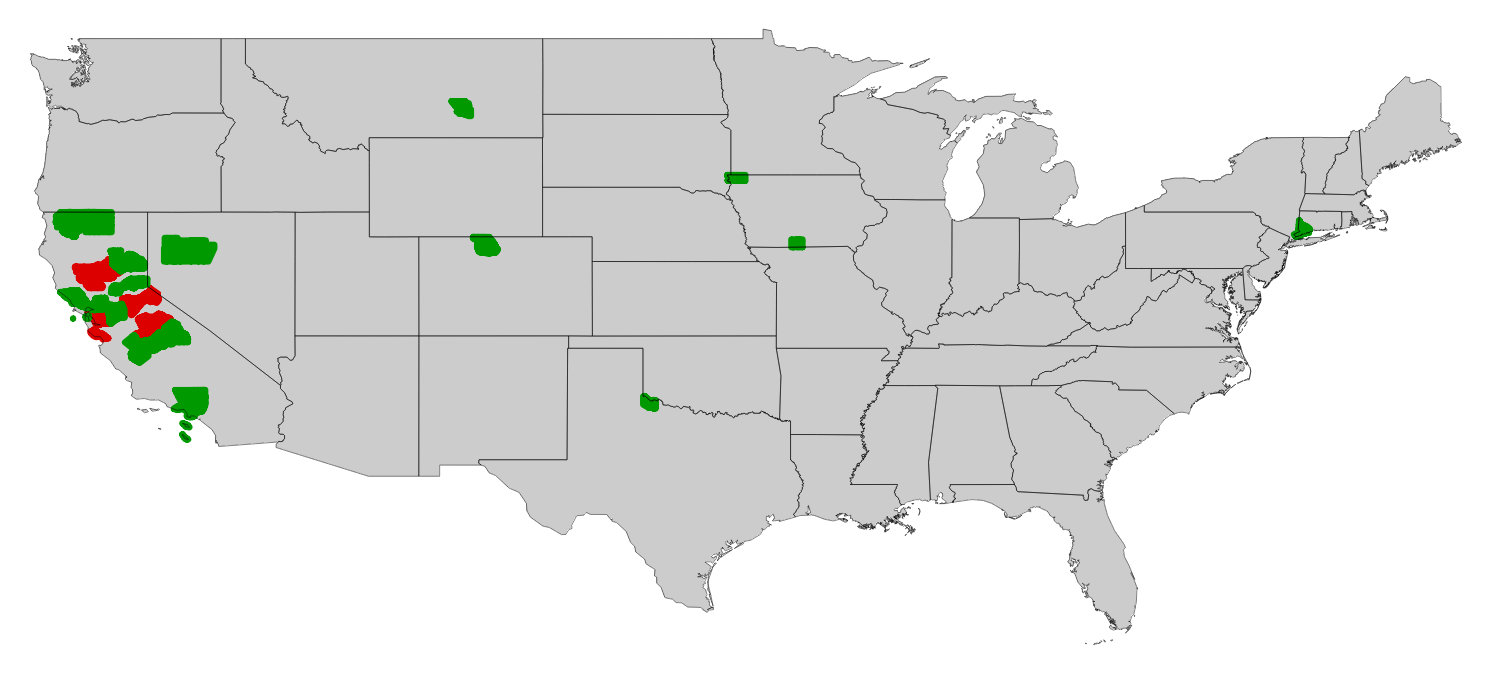}}
    \subfigure[Florida Seed] {\includegraphics[width=0.24\columnwidth]{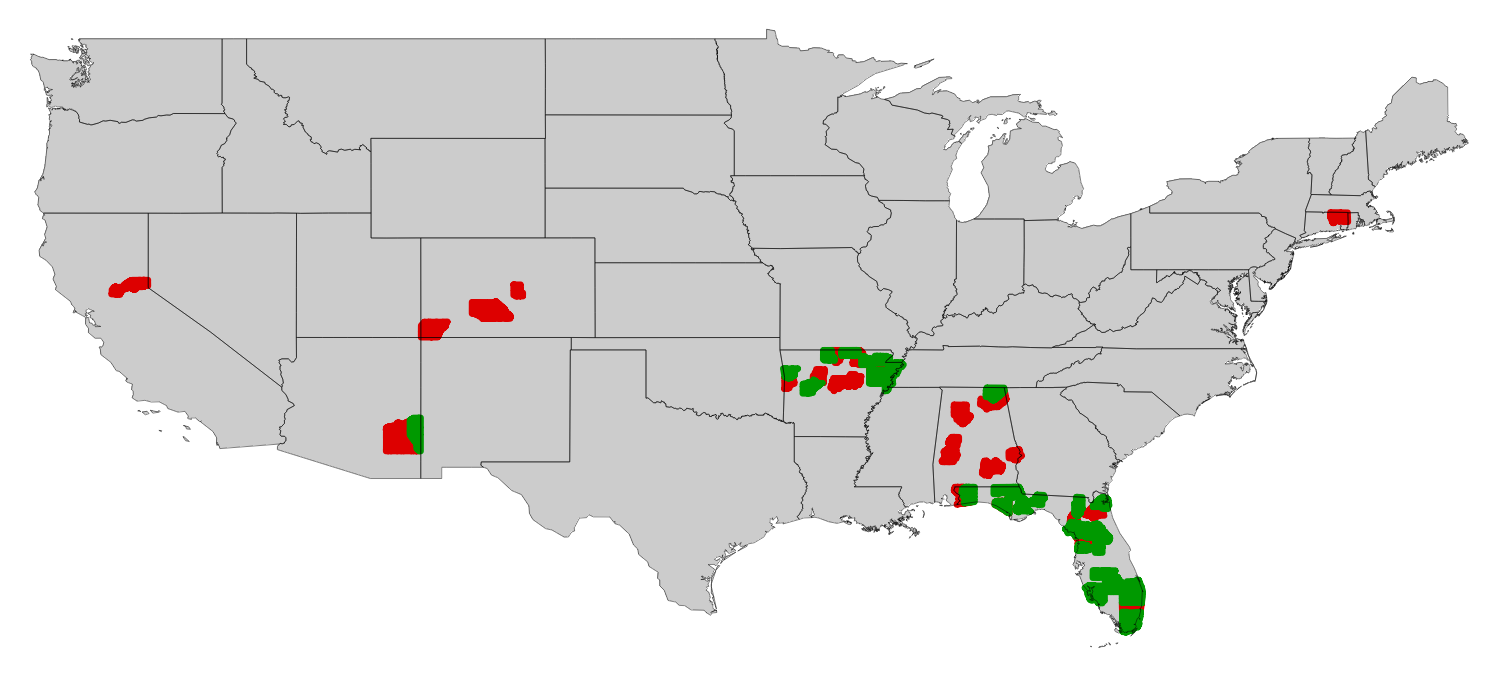}}
    \caption{\small
    (a)-(d) Clusters found by \texttt{LocBipartDC} on the interstate dispute network, using the USA as the seed vertex. In each case, countries in the yellow cluster tend to enter conflicts with countries in the blue cluster and vice versa.
    (e)-(h) Clusters found by \texttt{EvoCutDirected} in the US migration network. There is a large net flow from the red counties to the green counties.
    \label{fig:intro_examples}}
\end{figure*}

\section{Introduction} \label{sec:introduction}
 Given an arbitrary vertex $u$ of a graph $G=(V,E)$ as input, a local graph   clustering   algorithm  finds some   low-conductance set $S\subset V$ containing $u$, while the algorithm runs in time proportional to the size of the target cluster and independent of the size of the graph $G$. Because of the increasing size of available datasets,  which 
 makes centralised computation too expensive, local graph clustering has become an important learning technique for 
   analysing a number of large-scale graphs and has been  applied to solve many other learning and combinatorial optimisation  problems~\cite{Andersen2010,AGP+2016,fountoulakisPNormFlowDiffusion2020,liuStronglyLocalPnormcut2020,kdd/Takai0IY20,  wangCapacityReleasingDiffusion2017, YBL+2017, ZLM2013}. 

\subsection{Our Contribution}
 We study local graph clustering for learning the structure of clusters that are defined by their inter-connections,
 and present two local algorithms to achieve this objective in both undirected graphs and directed ones. 
 
Our first result is a local algorithm for  finding densely connected clusters
in an undirected graph $G=(V,E)$: given any seed vertex $u$, our algorithm is designed to find \emph{two} clusters $L, R$ around $u$, which are densely connected to each other and are loosely connected to $V\setminus (L\cup R)$. The design of our algorithm is based on a new reduction that allows us to relate the connections between $L,R$ and $V\setminus(L\cup R)$ to a \emph{single} cluster in the resulting  graph $H$, and a generalised analysis of Pagerank-based algorithms for local graph clustering. The significance of our designed algorithm is demonstrated by our experimental results on the Interstate Dispute Network from 1816 to 2010. By connecting two vertices~(countries) with an undirected edge
 weighted according to the severity of their military disputes
and  using the USA as the seed vertex,
our algorithm recovers two groups of countries that tend to have conflicts with each other, and shows how the two groups evolve over time. In particular, as shown in Figures~\ref{fig:intro_examples}(a)-(d), our algorithm not only identifies the changing roles of Russia, Japan, and eastern Europe in line with 20th century geopolitics, but also the reunification of east and west Germany around 1990.

We further study  densely connected clusters 
in a \emph{directed graph}~(digraph). Specifically, given any vertex $u$ in a digraph $G=(V,E)$ as input, our second  local algorithm outputs two disjoint vertex sets $L$ and $R$, such that (i) there are many edges \emph{from $L$  to $R$}, and (ii) $L\cup R$ is loosely connected to $V\setminus (L\cup R)$. The design of our algorithm is based on the following two   techniques: (1) a
new reduction that allows us to relate the edge weight from $L$ to $R$, as well as the edge connections between $L\cup R$ and $V\setminus (L \union R)$, to a \emph{single} vertex set in the resulting \emph{undirected} graph $H$; (2) a refined analysis of the \texttt{ESP}-based algorithm for local graph clustering.
 We show that our algorithm is able to recover  local densely connected clusters  in the US migration dataset, in which two vertex sets $L,R$ defined as above could represent a higher-order migration trend. In particular, as shown in  Figures~\ref{fig:intro_examples}(e)--(h), by using different counties as  starting vertices, our algorithm uncovers refined and more localised migration patterns than the previous work on the same dataset~\cite{CLS+2020}. To the best of our knowledge, our work represents  the first local clustering algorithm that achieves a similar goal.
 
\subsection{Related Work}
Li and Peng~\cite{LP2013} study the same problem for undirected graphs,
 and present a random walk based algorithm. Andersen~\cite{Andersen2010} studies a similar problem for undirected graphs under a different objective function, and our algorithm's runtime significantly improves the one in \cite{Andersen2010}.
 There is some recent work on local clustering for hypergraphs~\cite{kdd/Takai0IY20},
 and algorithms for finding higher-order structures of graphs based on network motifs both centrally~\cite{BensonGL15, Benson163} and locally~\cite{YBL+2017}.
 These algorithms are designed to find different types of clusters, and cannot be directly compared with ours.
Our problem is related to the problem of finding clusters in \emph{disassortative} networks~\cite{mooreActiveLearningNode2011, PWC+2019, zhu2020beyond}, although the existing techniques are based on semi-supervised, global methods while ours is unsupervised and local. 
There are also recent studies  which find clusters with a specific structure
in the centralised setting~\cite{CLS+2020,laenenHigherOrderSpectralClustering2020}.
Our current work shows that   such clusters can be learned locally via our presented new techniques.

\section{Preliminaries} \label{sec:prelim}

\paragraph{Notation.}
For any undirected and unweighted graph $G=(V_G, E_G)$ with $n$ vertices and $m$ edges, the   degree of any vertex $v$ is denoted by $\deg_G(v)$, and the set of neighbours of $v$ is $N_G(v)$.
For any   $S\subseteq V$, its volume is $\vol_G(S)\triangleq \sum_{v\in S} \deg(v)$, 
its boundary is $\partial_G (S)\triangleq \{ (u,v) \in E : u\in S~\mbox{and}~ v\not\in S \}$, and its conductance is
\[
\Phi_G(S)\triangleq \frac{|\partial_G(S)|}{\min\{\vol_G(S),\vol_G(V\setminus S)\}}.
\]
For disjoint $S_1, S_2 \subset V_G$, $e(S_1, S_2)$ is the number of edges between $S_1$ and $S_2$.
When $G=(V_G, E_G)$ is a digraph,  for any $u\in V_G$ we use $\deg_{\mathrm{out}}(u)$ and $\deg_{\mathrm{in}}(u)$ to express the number of edges with $u$ as the tail or the head, respectively. 
For any $S\subset V_G$, we define $
    \mathrm{vol}_{\mathrm{out}}(S)  = \sum_{u \in S} \deg_{\mathrm{out}}(u)$, and $
    \mathrm{vol}_{\mathrm{in}}(S)  = \sum_{u \in S} \deg_{\mathrm{in}}(u)$.

For undirected graphs, we use $D_G$ to denote the $n\times n$ diagonal matrix  with $(D_G)_{v,v} = \deg_G(v)$ for any vertex $v\in V$, and we use $A_G$ to represent the adjacency matrix of $G$ defined by $(A_G)_{u,v} = 1$ if $\{u,v\}\in E_G$, and $(A_G)_{u,v}=0$ otherwise. 
The lazy random walk matrix of $G$ is $W_G = (1/2)\cdot (I+ D^{-1}_G A_G)$.  For any set $S\subset V_G$, $\chi_S$ is the indicator vector of $S$, i.e., $\chi_S(v)=1$ if $v\in S$, and $\chi_S(v)=0$ otherwise.  If the set consists of a single vertex $v$, we write $\chi_v$ instead of $\chi_{\{v\}}$.
Sometimes we drop the subscript $G$ when the underlying graph is clear from the context.
 For any vectors $x,y\in\mathbb{R}^n$, we write $x\preceq y$ if it holds for any $v$ that $x(v)\leq y(v)$. For any operators $f, g:\mathbb{R}^n\rightarrow \mathbb{R}^n$,
 we define $f\circ g:\mathbb{R}^n\rightarrow \mathbb{R}^n$ by  $f\circ g(v) \triangleq f(g(v))$ for any $v$.
 For any vector $p$, we define the support of $p$ to be $\mathrm{supp}(p) = \{u : p(u) \neq 0\}$. The sweep sets of $p$ are defined by (1) ordering all the vertices  such that $$\frac{p(v_1)}{\deg(v_1)} \geq \frac{p(v_2)}{\deg(v_2)} \geq \ldots \geq \frac{p(v_n)}{\deg(v_n)},$$ and (2)  constructing $S_j^p = \{v_1, \ldots, v_j\}$ for $1\leq j\leq n$.
Throughout this paper, we will consider vectors to be row vectors, so the random walk update step for a distribution $p$ is written as $pW$.
For ease of presentation we consider only unweighted graphs; however, our analysis can be easily generalised to the weighted case.
 The omitted proofs can be found in the Appendix.

\paragraph{Pagerank.} 
Given an underlying graph $G$ with the lazy random walk matrix $W$, the personalised Pagerank vector $\ppr(\alpha, s)$ is defined to be the unique solution of the equation
\begin{equation}\label{eq:defpr}
    \ppr(\alpha, s) = \alpha s + (1 - \alpha) \ppr(\alpha, s) W,
\end{equation}
where $s \in \R^{n}_{\geq 0}$ is a starting vector and $\alpha\in(0,1]$ is called the teleport probability. 
Andersen et al.~\cite{ACL2006} show  that  the personalised Pagerank vector can be written as 
$   \pr(\alpha, s) = \alpha \sum_{t = 0}^\infty (1 - \alpha)^t s W^t.
$
Therefore, we could study $\pr(\alpha, s)$ through the following random process:   pick some integer $t\in\mathbb{Z}_{\geq 0}$ with probability $\alpha(1-\alpha)^t$, and   perform a $t$-step lazy random walk, where the starting vertex of the random walk is picked according to $s$. Then, $\pr(\alpha, s)$   describes the probability of reaching each vertex in this process. 
  
Computing an exact Pagerank vector $\pr(\alpha, \chi_v)$ is equivalent to computing the stationary distribution of a Markov chain on the vertex set $V$
which
has a time complexity of $\Omega(n)$. 
However, since the probability mass of a personalised Pagerank vector is concentrated around some starting vertex, it is possible to compute a good approximation of the Pagerank vector in a local way.
Andersen et al.~\cite{ACL2006}
introduced the approximate Pagerank, which will be used in our analysis.
\begin{definition}
    A vector $p = \apr(\alpha, s, r)$ is an approximate Pagerank vector if 
    $p + \pr(\alpha, r) = \pr(\alpha, s)$.
    The vector $r$ is called the residual vector.
\end{definition}

\paragraph{The evolving set process.}
     The evolving set process (ESP) is a Markov chain whose states are sets of vertices $S_i \subseteq V$.
    Given a state $S_i$, the next state $S_{i+1}$ is determined by the following process: (1) choose $t \in [0, 1]$ uniformly at random; (2) let $S_{i + 1} = \{v \in V | \chi_v W \chi_{S_i}^\transpose \geq t\}$. 
    The volume-biased ESP is a variant used to ensure that the Markov chain absorbs in the state $V$ rather than $\emptyset$.
    Andersen and Peres~\cite{AC2009} gave a local   algorithm for undirected graph clustering using the volume-biased ESP.
    In particular, they gave an algorithm  $\texttt{GenerateSample}(u, T)$ which samples the $T$-th element from the volume-biased ESP with $S_0 = \{u\}$.

\section{The Algorithm for Undirected Graphs\label{sec:undirected}}
Now we present a local algorithm for finding two clusters in an undirected graph with a dense cut between them. To formalise this notion, for any  undirected graph $G=(V,E)$ and disjoint $L,R\subset V$, we follow Trevisan~\cite{Trevisan2012} and define the \emph{bipartiteness} ratio  
\[
\beta(L,R) \triangleq 1 - \frac{2e(L,R)}{\vol(L\cup R)}.
\]
Notice that a low $\beta(L,R)$ value means that there is a dense cut between $L$ and $R$, and there is a sparse cut between 
$L\cup R$ and $V\setminus (L\cup R)$.  In particular,  $\beta(L,R)=0$ implies that $(L,R)$ forms a bipartite  and connected component of $G$.
We will describe a local algorithm for finding almost-bipartite sets $L$ and $R$ with a low value of $\beta(L, R)$.



\subsection{The Reduction by Double Cover} \label{sec:reduction}
The design of most local algorithms for finding  a target set $S\subset V$ of low conductance is  based on analysing the behaviour of random walks starting from vertices in $S$. In particular, when the conductance $\Phi_G(S)$ is low, a random walk starting from most vertices in $S$ will leave $S$ with low probability.
However, for our setting, the target is a pair of sets $L, R$ with many connections between them. As such, a random walk starting in either $L$ or $R$ is very likely to leave the starting set. To address this,
 we introduce a novel technique  based on the double cover of $G$
to reduce the problem of finding two sets of high conductance to the problem of finding one of \emph{low} conductance.

Formally, for any undirected graph $G=(V_G, E_G)$, its double cover is the  bipartite graph $H=(V_H, E_H)$ defined as follows: (1) every vertex $v\in V_G$ has two corresponding vertices $v_1, v_2\in V_H$; (2) for every edge $\{u,v\}\in E_G$, there are edges $\{u_1, v_2\}$  and $\{u_2, v_1\}$ in $E_H$. See Figure~\ref{fig:dc} for an illustration.

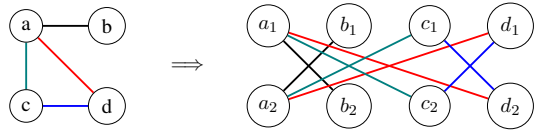
\begin{figure}[h]
\centering
 \resizebox{3\totalheight}{!}{
\input{dcconstruct2.tikz}}
\caption{\small{An example of the construction of the double cover.\label{fig:dc}}}
\end{figure}

 Now  we  present a tight  connection between the value of $\beta(L,R)$ for any disjoint  sets  $L, R\subset V_G$  and
 the conductance of
 a \emph{single} set in the double cover of $G$.  To this end, for any $S\subset V_G$, we define $S_1\subset V_H$ and $S_2\subset V_H$ by $S_1\triangleq \{v_1~|~ v\in S\}$ and $S_2\triangleq \{v_2~|~ v\in S\}$.
 We formalise the connection in the following lemma.
 
\begin{lemma} \label{lem:cond_bip}
 Let $G$ be an undirected graph, $S\subset V$ with partitioning $(L, R)$, and $H$ be the double cover of $G$.  Then, it holds that
    $
        \cond[H]{L_1 \cup R_2} = \bipart[G]{L}{R}$.
\end{lemma}

Next we   look at  the other direction of this correspondence.
Specifically, given any  $S \subset V_H$ in the double cover of a graph $G$, we would like to find two disjoint sets $L \subset V_G$ and $R \subset V_G$ such that $\bipart[G]{L}{R} = \cond[H]{S}$.
However, such a connection does not hold in general.
To overcome this, we restrict our attention to those subsets of $V_H$ which can be unambiguously interpreted as two disjoint sets in $V_G$.

\begin{definition} \label{def:simple}
We call $S\subset V_H$ \textbf{simple} if  $| \{v_1,v_2\}\cap S|\leq 1$ holds for all $v\in V_G$.
\end{definition}

\begin{lemma}\label{lem:ReductionForSimpleset}
    For any  simple set $S \subset V_{H}$, let  $L = \{u : u_1 \in S\}$ and $R = \{u : u_2 \in S\}$. Then,   
    $
        \bipart[G]{L}{R} = \cond[H]{S}.
    $
\end{lemma}

\subsection{Design of the Algorithm} \label{sec:algdesc}
So far we have shown that the problem of finding densely connected sets 
$L, R\subset V_G$ can be reduced to finding $S\subset V_H$ of low conductance in the double cover $H$, and this reduction raises the natural question of whether existing local algorithms can  be directly  employed to find  $L$ and $R$ in $G$.
However, this is not the case: even though a set $S\subset V_H$ returned by most local algorithms is guaranteed to have low conductance, vertices of $G$ could be included in $S$ twice, and as such $S$ will not necessarily give us disjoint sets $L, R\subset V_G$ with low value of $\beta(L,R)$.
See Figure~\ref{fig:nonsimple} for an illustration.


\begin{figure}[h]
\centering
\resizebox{2\totalheight}{!}{
\input{nonsimple2.tikz}}
\caption{\small{The indicated vertex set $S\subset V_H$
has low conductance.
However, since $S$  contains  many pairs of vertices which correspond to the same vertex in $G$, $S$ gives us little information for finding disjoint $L,R\subset V_G$ with a low $\beta(L,R)$ value.}}
\label{fig:nonsimple}
\end{figure}
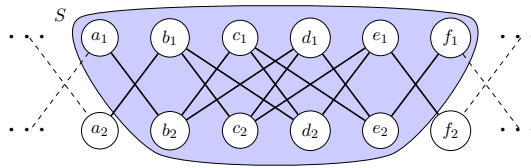

\paragraph{The simplify operator.}
To take the example shown in Figure~\ref{fig:nonsimple} into account,  our objective is to  design a local algorithm for finding some set $S\subset V_H$ of low conductance which is also simple.
To ensure this, we  introduce  the simplify operator, and analyse its properties.

\begin{definition}[Simplify operator] Let $H$ be the double cover of $G$, where the two corresponding vertices of any $u\in V_G$ are defined as $u_1,u_2\in V_H$. Then, for any   $p\in\mathbb{R}^{2n}_{\geq 0}$,  the simplify operator is a function $\sigma:\mathbb{R}^{2n}_{\geq 0}\rightarrow \mathbb{R}^{2n}_{\geq 0}$  defined by 
\begin{align*}
    	(\sigma\circ p)(u_1) & \triangleq \max (0, p(u_1) - p(u_2)), \\
        (\sigma\circ p)(u_2) & \triangleq \max (0, p(u_2) - p(u_1))
\end{align*}
for  every $u\in V_G$.
\end{definition}
 
Notice that, for any vector $p$ and any $u\in V_G$,  at most one of $u_1$ and $u_2$ is in the support of $\sigma \circ p$; hence, the support of $\sigma\circ p$ is always simple. To explain the meaning of $\sigma$, for a vector $p$ one could view $p(u_1)$ as our ``confidence'' that $u\in L$, and   $p(u_2)$ as our ``confidence'' that $u\in R $. Hence, when $p(u_1) \approx p(u_2)$, both $(\sigma \circ p)(u_1)$ and $(\sigma \circ p)(u_2)$ are small which captures the fact that we would not prefer to include $u$ in either $L$ or $R$. 
On the other hand, when $p(u_1) \gg p(u_2)$, we have $(\sigma \circ p)(u_1) \approx p(u_1)$, which captures our confidence that $u$ should belong to  $L$. 
The following lemma summaries some key properties of $\sigma$.

\begin{lemma} \label{lem:sigmaproperty}
The  following  holds for the 
$\sigma$-operator:   
\begin{itemize} 
    \item $\sigma\circ(c\cdot p) = c\cdot (\sigma \circ p)$ for $p \in \R_{\geq 0}^{2n}$ and any $c\in\R_{\geq 0}$;
    \item $\sigma\circ(a + b) \preceq \sigma\circ a+ \sigma\circ b$ for $a, b \in \mathbb{R}_{\geq 0}^{2n}$;
    \item $\sigma \circ (pW) \preceq (\sigma \circ p)W$ for $p \in \mathbb{R}_{\geq 0}^{2n}$.
\end{itemize}
\end{lemma}

While these three properties will all be used in our analysis, the third is of particular importance:  it implies that, if $p$ is the probability distribution of a random walk in $H$, applying $\sigma$ before taking a one-step  random walk would never result in lower probability mass than applying $\sigma$ after taking a one-step random walk.
This means that no probability mass would be lost when the $\sigma$-operator is applied between every step of a random walk, in comparison with applying $\sigma$ at the end of an entire random walk process.

\paragraph{Description of the algorithm.}  Our proposed algorithm is conceptually   simple: every vertex $u$ of the input graph $G$ maintains two copies $u_1, u_2$ of itself, and these two ``virtual'' vertices are used to simulate $u$'s corresponding vertices in the double cover $H$ of $G$. Then, as the neighbours of $u_1$ and $u_2$ in $H$ are entirely determined by $u$'s neighbours in $G$ and can be constructed locally, a random walk process in $H$ will be simulated in $G$. This  will allow us to apply a local algorithm  similar to the one by Anderson et al.~\cite{ACL2006} on this ``virtual'' graph $H$.
Finally, since all the required information about $u_1,u_2\in V_H$ is maintained by $u\in V_G$, the $\sigma$-operator will be applied locally. 

The formal description of our algorithm is given in Algorithm~\ref{alg:local_max_cut}, which invokes Algorithm~\ref{alg:aprdc}  as the key component  to compute  $\apr_H(\alpha, \chi_{u_1}, r)$. Specifically,  Algorithm~\ref{alg:aprdc} maintains, for every vertex $u \in G$, tuples $(p(u_1), p(u_2))$ and $(r(u_1), r(u_2))$ to keep track of the values of $p$ and $r$ in $G$'s double cover.
    For a given vertex $u$, the entries in these tuples are expressed by $p_1(u), p_2(u), r_1(u),$ and $r_2(u)$ respectively.   Every $\texttt{dcpush}$ operation~(Algorithm~\ref{alg:dcpush}) preserves the invariant
   $
        p + \pr_H(\alpha, r) = \pr_H(\alpha, \chi_{v_1}),
$    which ensures that the final output of Algorithm~\ref{alg:aprdc} is an approximate Pagerank vector.
    We remark that, although   the presentation  of the \texttt{ApproximatePagerankDC} procedure is   similar to the one in Andersen et al.~\cite{ACL2006}, in our  \texttt{dcpush} procedure the update of the residual vector $r$ is slightly more involved: specifically, for every vertex $u\in V_G$, both $r_1(u)$ and $r_2(u)$ are needed in order to update $r_1(u)$ (or $r_2(u)$). That is one of the reasons that the performance of the algorithm in Andersen et al.~\cite{ACL2006} cannot be directly applied for our algorithm, and a more technical analysis, some of which is parallel to theirs, is needed in order to analyse the correctness and performance of our algorithm.

\begin{algorithm}[h]
   \caption{\texttt{LocBipartDC}\label{alg:local_max_cut}}
\begin{algorithmic}
   \STATE {\bfseries Input:} A graph $G$, starting vertex $u$, target volume $\gamma$, and target bipartiteness $\beta$
   \STATE {\bfseries Output:} Two disjoint sets $L$ and $R$
  \STATE Set $\alpha = \frac{\beta^2}{378}$, and  $\epsilon = \frac{1}{20 \gamma}$ \\
Compute $p' = \texttt{ApproximatePagerankDC}(u, \alpha, \epsilon)$ 
\STATE Compute $p = \sigma \circ p'$ 
\FOR{$j \in [1, |\mathrm{supp}(p)|]$}
    \IF{$\cond{\pjsweep} \leq \beta$}
      \STATE Set $L = \{u : u_1 \in \pjsweep \}$, and   $R = \{u : u_2 \in \pjsweep \}$ \\
    \Return $(L, R)$
    \ENDIF
\ENDFOR
\end{algorithmic}
\end{algorithm}

\begin{algorithm}[h]
   \caption{\texttt{ApproximatePagerankDC} \label{alg:aprdc}}
\begin{algorithmic}
   \STATE {\bfseries Input:} Starting vertex $v$,    parameters $\alpha$ and $\epsilon$ 
   \STATE {\bfseries Output:} Approximate Pagerank vector $\apr_H(\alpha, \chi_{v_1}, r)$
   \STATE Set $p_1 = p_2 = r_2 = \mathbf{0}$; set $r_1 = \chi_v$ 
   \WHILE{$\max_{(u, i) \in V \times \{1, 2\}} \frac{r_i(u)}{\deg(u)} \geq \epsilon$}
   \STATE Choose any $u$ and $i \in \{1, 2\}$ such that $\frac{r_i(u)}{\deg(u)} \geq \epsilon$
   \STATE  $(p_1, p_2, r_1, r_2) = \texttt{dcpush}(\alpha, (u, i), p_1, p_2, r_1, r_2)$
   \ENDWHILE
   
   \Return $p = \left[p_1, p_2 \right]$
\end{algorithmic}
\end{algorithm}


\begin{algorithm}[h]
   \caption{\texttt{dcpush}}
   \label{alg:dcpush}
\begin{algorithmic}
   \STATE {\bfseries Input:} $\alpha, (u, i), p_1, p_2, r_1, r_2$
   \STATE {\bfseries Output:} $(p'_1, p'_2, r'_1, r'_2)$
   \STATE Set $(p_1', p_2', r_1', r_2') = (p_1, p_2, r_1, r_2)$
   \STATE
Set $p'_i(u) = p_i(u) + \alpha r_i(u)$; $r'_i(u) = (1 - \alpha) \frac{r_i(u)}{2}$ \\
    \FOR{$v\in N_G(u)$}
   \STATE Set $r'_{3 - i}(v) = r_{3 - i}(v) + (1 - \alpha) \frac{r_i(u)}{2 \deg(u)}$
   \ENDFOR
   
   \Return $(p'_1, p'_2, r'_1, r'_2)$
 \end{algorithmic}
\end{algorithm}

\subsection{Analysis of the Algorithm} \label{sec:alganalysis}

     To prove the correctness of our algorithm, we will show two complementary facts which we state informally here: 
    \begin{enumerate}
        \item If there is a simple set $S \subset V_H$ with low conductance, then for most $u_1 \in S$, the simplified approximate Pagerank vector $p = \sigma \circ \apr(\alpha, \chi_{u_1}, r)$ will have a lot of probability mass on a small set of vertices.
        \item If $p = \sigma \circ \apr(\alpha, \chi_{u_1}, r)$ contains a lot of probability mass on some small set of vertices, then there is a sweep set of $p$ with low conductance.
    \end{enumerate}
    \vspace{-0.2cm}
        As we have shown in Section~\ref{sec:reduction}, there is a direct correspondence between almost-bipartite sets in $G$, and low-conductance and  simple sets in $H$.
    This means that the two facts above are exactly what we need to prove that
    Algorithm~\ref{alg:local_max_cut}
   can find densely connected sets in $G$.
  
We will first show in   Lemma~\ref{lem:sapr_escapingmass}  how the $\sigma$-operator affects some standard mixing properties of Pagerank vectors in order to establish the first fact promised above.
This lemma relies on the fact that $S \subset V_G$ corresponds to a \emph{simple} set in $V_H$.
This allows us to apply the $\sigma$-operator to the approximate Pagerank vector $\apr_H(\alpha, \chi_{u_1}, r)$ while preserving a large probability mass on the target set.
 
\begin{lemma} \label{lem:sapr_escapingmass}
 For any set $S \subset V_{G}$ with partitioning $(L, R)$ and any constant  $\alpha \in [0, 1]$, there is a subset $S_{\alpha} \subseteq S$ with $\vol(S_{\alpha}) \geq \vol(S) / 2$ such that, for any vertex $v \in S_{\alpha}$, the simplified approximate Pagerank on the double cover $p = \simplify{\apr_{H}(\alpha, \chi_{v_1}, r)}$ satisfies
    \[
    	p(L_1 \union R_2) \geq 1 - \frac{2 \bipart{L}{R}}{\alpha} - 2 \vol(S) \max_{u \in V} \frac{r(u)}{\deg(u)}.
    \]
\end{lemma}

 To  prove the second fact, we show as an intermediate lemma that the value of $p(u_1)$ can be bounded with respect to its value after taking a step of the random walk: $pW(u_1)$. 
  
\begin{lemma} \label{lem:sapr_updatestep} Let $G$ be a  graph with double cover $H$, and $\apr(\alpha, s, r)$ be the approximate Pagerank vector defined with respect to $H$. Then, $p = \simplify{\apr(\alpha, s, r)}$ satisfies that
$ p(u_1)  \leq \alpha \left(s(u_1) + r(u_2)\right) + (1 - \alpha)(p W)(u_1)$, and $ p(u_2)  \leq \alpha \left(s(u_2) + r(u_1)\right) + (1 - \alpha)(p W)(u_2)$
for any $u\in V_G$.
\end{lemma}
 Notice that applying the $\sigma$-operator for any vertex $u_1$ introduces a new dependency on the value of the residual vector $r$ at $u_2$.
This subtle observation demonstrates the additional complexity introduced by the $\sigma$-operator when compared with previous analysis of Pagerank-based local algorithms~\cite{ACL2006}.
 Taking account of the $\sigma$-operator, we further analyse the  Lov\'asz-Simonovits curve defined by $p$, which is a common technique in the analysis of random  walks on graphs~\cite{LS1990}: we show that if there is a set $S$ with a large value of $p(S)$, there must be a sweep set $S_j^p$ with small conductance.

\begin{lemma} \label{lem:probimpliescond}
 Let $G$ be a graph with double cover $H$, and let $p = \simplify{\apr_H(\alpha, s, r)}$ such that  $\max_{u \in V_H} \frac{r(u)}{d(u)} \leq \epsilon$.
    If there is a set of vertices $S \subset V_H$ and a constant $\delta$ such that
    $
        p(S) - \frac{\vol(S)}{\vol(V_H)} \geq \delta$,
    then there is some $j \in [1, \abs{\mathrm{supp}(p)}]$ such that
    $
        \cond[H]{\pjsweep} < 6 \sqrt{ (1 + \epsilon \vol(S)) \alpha \ln (\frac{4}{\delta})\big/\delta}$.
\end{lemma}
    We have now shown the two facts promised at the beginning of this subsection.
    Putting these together,
     if there is a simple set $S \subset V_H$ with low conductance then we can find a sweep set of $\sigma \circ \apr(\alpha, s, r)$ with low conductance.
    By the reduction from almost-bipartite sets in $G$ to low-conductance simple sets in $H$  our target set corresponds to a simple set $S \subset V_H$ which leads to Algorithm~\ref{alg:local_max_cut} for finding almost-bipartite sets. Our result is summarised as follows.

\begin{theorem} \label{thm:main_thm}
 Let $G$ be an $n$-vertex undirected graph, and $L, R \subset V_G$ be disjoint sets such that $\bipart{L}{R} \leq \beta$ and $\vol(L \union R) \leq \gamma$. Then, there is a set $C \subseteq L \union R$ with $\vol(C)\geq \vol(L \union R)/2$ such that, for any $v\in C$,  $\texttt{LocBipartDC}\left(G, v, \gamma, \sqrt{7560 \beta}\right)$ returns $(L', R')$ with $\bipart{L'}{R'} = \bigo{\sqrt{\beta}}$ and  ${\vol(L'\union R')} = \bigo{\beta^{-1}\gamma}$.
Moreover, the algorithm has running time $\bigo{\beta^{-1}\gamma \log n}$. 
\end{theorem}

The quadratic approximation guarantee in Theorem~\ref{thm:main_thm} matches the state-of-the-art local algorithm for finding a single set with low conductance~\cite{AGP+2016}.
Furthermore,
 our result presents a significant improvement over the previous state-of-the-art by Li and Peng~\cite{LP2013}, whose design is based on an entirely different technique than ours. For any $\epsilon \in [0, 1/2]$, their algorithm runs in time $\bigo{\epsilon^2 \beta^{-2}\gamma^{1 + \epsilon}\log^3\gamma}$ and returns a set with volume $\bigo{\gamma^{1 + \epsilon}}$ and bipartiteness ratio $\bigo{\sqrt{\beta / \epsilon}}$. In particular, their algorithm requires much higher time complexity in order to guarantee the same bipartiteness ratio $O\left(\sqrt{\beta}\right)$.
 

\section{The Algorithm   for Digraphs} \label{sec:directed}
We now turn our attention to
local algorithms for
 finding densely-connected clusters
in digraphs.
In comparison with undirected graphs, we are interested in   finding disjoint   $L,R \subset V$  of some digraph $G=(V,E)$  such that most of the edges adjacent to  $L\cup R$ are \emph{from $L$ to $R$}. To formalise this,
we define the \emph{flow ratio} from $L$ to $R$ as \[
 F(L,R) \triangleq 1  - \frac{2e(L,R)}{\vol_{\mathrm{out}}(L) + \vol_{\mathrm{in}}(R)},
 \]
 where $e(L,R)$ is the number of directed edges from $L$ to $R$. Notice that  we take not only edge densities but also edge directions into account:
 a low   $F(L,R)$-value tells us that almost all edges with their tail in $L$ have their head in $R$, and conversely almost all edges with their head in $R$ have their tail in $L$. One could also  see $F(L, R)$ as a generalisation of $\beta(L,R)$.
 In particular, if we view an undirected graph as a digraph by replacing each edge with two directed edges, then
 $\beta(L, R) = F(L, R)$. In this section, we will present a local algorithm for finding   such vertex sets in a digraph, and analyse the algorithm's performance.

\subsection{The Reduction  by Semi-Double Cover}

Given  a digraph $G=(V_G,E_G)$, we construct its \emph{semi-double cover} $H=(V_H, E_H)$ as follows:
(1) every vertex $v\in V_G$ has two corresponding vertices $v_1, v_2\in V_H$; (2) for every edge $(u,v)\in E_G$, we add the edge $\{u_1, v_2\}$ in $E_H$, see Figure~\ref{fig:directedconstruct} for illustration. 
\footnote{We remark that this reduction was also used by Anderson~\cite{Andersen2010} for finding dense components in a digraph.}
It is worth comparing this reduction with the one for undirected graphs:  
\begin{itemize} 
    \item For undirected graphs, we apply the standard double cover and every undirected edge in $G$ corresponds to two edges in the double cover $H$;
    \item For digraphs, every directed edge in $G$ corresponds to \emph{one} undirected edge in  $H$. This \emph{asymmetry} would allow us to ``recover'' the direction of any edge in $G$.
\end{itemize}
\vspace{-0.2cm}

 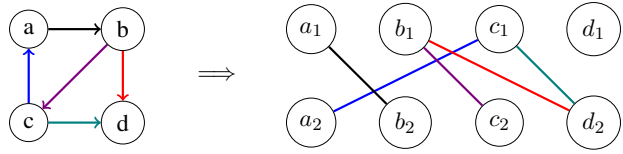
\begin{figure}[h]
     \centering
 \resizebox{3.5\totalheight}{!}{    \input{dcconstruct.tikz}}
     \caption{\small{An example of  a digraph  and its semi-double cover.}}
     \label{fig:directedconstruct}
 \end{figure}

We follow the use of $S_1, S_2$ from Section~\ref{sec:undirected}:  
for any $S\subset V_G$,
we define $S_1\subset V_H$ and $S_2\subset V_H$ by $S_1\triangleq \{v_1~|~ v\in S\}$ and $S_2\triangleq \{v_2~|~ v\in S\}$. 
The lemma below shows the connection between the value of $F_G(L,R)$ for any   $L,R$ and $\Phi_H(L_1\cup R_2)$.
 
\begin{lemma} \label{lem:flow_cond}
Let $G$ be a digraph with semi-double cover $H$. Then, it holds  for any $L, R \subset V_G$
that 
   $
        F_G(L, R) = \cond[H]{L_1 \union R_2}$.
    Similarly, for any simple set $S \subset V_H$, let $L = \{u : u_1 \in S\}$ and $R = \{u : u_2 \in S\}$. Then, it holds that $F_G(L, R) = \Phi_H(S)$.
\end{lemma}

\subsection{Design and Analysis of the Algorithm}
 Our presented algorithm is a modification of the algorithm by Andersen and Peres~\cite{AC2009}.
 Given a digraph $G$ as input, our algorithm simulates the volume-biased ESP on $G$'s semi-double cover $H$. Notice that the graph $H$ can be constructed locally in the same way as the local construction of the double cover. However, as the output set 
$S$ of an ESP-based algorithm is not necessarily simple, our algorithm only returns vertices $u\in V_G$ in which \emph{exactly} one of $u_1$ and $u_2$ are included in $S$.
The key procedure for our algorithm is given in Algorithm~\ref{alg:directed_evo_cut},   in which the \texttt{GenerateSample} procedure is the one described at the end of Section~\ref{sec:prelim}.

\begin{algorithm}[h]
   \caption{\texttt{EvoCutDirected}~(\texttt{ECD})}
   \label{alg:directed_evo_cut}
\begin{algorithmic}
   \STATE {\bfseries Input:} Starting vertex $u$, $i \in \{1, 2\}$,  target flow ratio $\phi$
   \STATE {\bfseries Output:} A pair of sets $L, R\subset V_G$
   \STATE Set $T = \lfloor (100 \phi^{\frac{2}{3}})^{-1}\rfloor$. 
   \STATE Compute $S = \texttt{GenerateSample}_H(u_i, T)$
   \STATE Let $L = \{u \in V_G : u_1 \in S$ and $u_2 \not \in S\}$
   \STATE Let $R = \{u \in V_G : u_2 \in S$ and $u_1 \not \in S\}$ 
   
   \Return $L$ and $R$
\end{algorithmic}
\end{algorithm}

Notice that in our constructed graph $H$,
$\Phi(L_1 \union R_2) \leq \phi$ does not imply that $\Phi(L_2 \union R_1) \leq \phi$.
Due to this asymmetry, Algorithm~\ref{alg:directed_evo_cut} takes a parameter $i \in \{1, 2\}$ to indicate whether the starting vertex is in $L$ or $R$.
If it is not known whether $u$ is in $L$ or $R$,
two copies
can be run in parallel, one with $i = 1$ and the other with $i = 2$.
Once one of them  terminates with the performance guaranteed   in Theorem~\ref{thm:directedresult}, the other
can be terminated. 
Hence,
we can  always assume that it is known whether the starting vertex $u$ is in $L$ or $R$.



Now we sketch the analysis of the algorithm. Notice that, since  the evolving set process gives us an arbitrary set on the semi-double cover,  in Algorithm~\ref{alg:directed_evo_cut} we convert this into a simple set by 
removing any vertices $u$ where $u_1 \in S$ and $u_2 \in S$. The following definition allows us to discuss sets which are close to being simple.
\begin{definition}[$\epsilon$-simple set] For any set $S\subset V_H$, let $P=\{u_1,u_2: u\in V_G, u_1\in S~\mbox{and}~u_2\in S\}$. We call set $S$ $\epsilon$-simple if it holds that
   $
        \frac{\vol(P)}{\vol(S)} \leq \epsilon.
    $
\end{definition}
 The notion of $\epsilon$-simple sets measures the ratio of vertices
 in which both $u_1$ and $u_2$ are in $S$.
 In particular,  any simple set defined in Definition~\ref{def:simple} is $0$-simple.
 We show  that, for any $\epsilon$-simple set $S\subset V_H$, one can construct a simple set $S'$ such that $\cond{S'} \leq \frac{1}{1 - \epsilon}\cdot  \left(\cond{S} + \epsilon\right)$.
Therefore, in order to guarantee that $\Phi(S')$ is small, we need to construct $S$ such that $\Phi(S)$ is small and $S$ is $\epsilon$-simple for small $\epsilon$.
 Because of this,
 our presented algorithm uses a lower value of $T$ than the algorithm  in Andersen and Peres~\cite{AC2009}; this allows us to better control  $\vol(S)$ at the cost of a slightly worse approximation guarantee. Our algorithm's performance is summarised   in Theorem~\ref{thm:directedresult}. 

\begin{theorem} \label{thm:directedresult}
 Let $G$ be an $n$-vertex digraph, and $L, R \subset V_G$ be disjoint sets such that $F(L, R) \leq \phi$ and $\vol(L \cup R)\leq \gamma$. There is a set $C \subseteq L \cup R$ with $\vol(C) \geq \vol(L \cup R)/2$ such that, for any $v\in C$ and some $i \in \{1, 2\}$,
$\texttt{EvoCutDirected}(G, v, i, \phi)$ returns $(L', R')$ such that  $F(L', R') = \bigo{\phi^{\frac{1}{3}} \log^{\frac{1}{2}} n }$ and  ${\vol(L'\union R')} =O\left( (1 - \phi^{\frac{1}{3}})^{-1} \gamma\right)$.
Moreover, the algorithm has running time $\bigo{ \phi^{-\frac{1}{2}}\gamma \log^{\frac{3}{2}} n}$.
\end{theorem}

To the best of our knowledge, this is the first local algorithm for digraphs that approximates a pair of densely connected clusters,
and demonstrates that finding such a pair appears to be much easier than finding a low-conductance set in a digraph; in particular, existing local algorithms for finding a low-conductance set require the stationary distribution of the random walk in the digraph~\cite{AC2007}, the sublinear-time computation of which is unknown~\cite{CKP+2017}. However, knowledge of the stationary distribution is not needed for our algorithm.



 
\paragraph{Further Discussion.}
It is important to note that the semi-double cover construction is able to handle directed graphs which contain edges between two vertices $u$ and $v$ in both directions.
In other words, the adjacency matrix of the digraph need not be skew-symmetric.
This is an advantage of our approach over previous methods~(e.g., \cite{CLS+2020}), and it would be a meaningful research direction to identify the benefit this gives our developed reduction.

It is also insightful to discuss why Algorithm~\ref{alg:local_max_cut} cannot be applied for digraphs, although the input digraph is translated into an \emph{undirected} graph by our reduction.
This is because, when translating a digraph into a bipartite undirected graph, the third property of the $\sigma$-operator in Lemma~\ref{lem:sigmaproperty} no longer holds, since the existence of the edge $\{u_1, v_2\}\in E_H$ does not necessarily imply that $\{u_2, v_1\}\in E_H$. Indeed,
Figure~\ref{fig:counterexample} gives a counterexample in which $\left(\sigma \circ (p W)\right)(u) \not \leq \left((\sigma \circ p)W\right)(u)$.
This means that the typical analysis of a Pagerank vector with the Lov\'asz-Simonovitz curve cannot be applied anymore.
In our point of view, constructing some operator similar to our $\sigma$-operator and designing a  Pagerank-based local algorithm for digraphs based on such an operator is a very interesting open question, and may help to close the gap in the approximation guarantee between the undirected and directed cases.
 
\begin{figure}[h] 
    \centering
    \input{counterexample.tikz}
    \caption{\small Consider the
    digraph and its semi-double cover above. Suppose $p(a_1) = p(a_2) = 0.5$ and $p(b_1) = p(b_2) = 0$. It is straightforward to check that $(\sigma \circ (p W))(b_2) = 0.25$ and $((\sigma \circ p) W)(b_2) = 0$.}
    \label{fig:counterexample}
\end{figure}
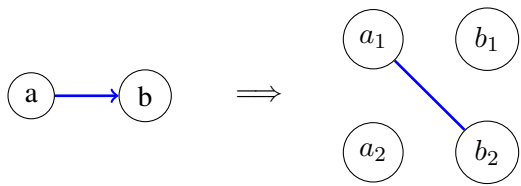

 In addition,
we  underline that  one  cannot apply the tighter analysis of the ESP process given by Andersen et al.\ \cite{AGP+2016} to our algorithm.
The key to their analysis is an
improved bound on the probability that a random walk escapes from the target cluster.
In order to take advantage of this, they use a larger value of $T$ in the algorithm which relaxes the guarantee on the volume of the output set.
Since our analysis relies on a very tight guarantee on the overlap of the output set with the target set, we cannot use their improvement in our setting.



\section{Experiments} \label{sec:experiments}

In this section we evaluate the performance of our proposed algorithms on both synthetic and real-world data sets. 
For undirected graphs, we compare the performance of our algorithm against the previous state-of-the-art~\cite{LP2013}, through the synthetic dataset with various parameters and apply the real-world dataset to demonstrate the significance of our algorithm.
For directed graphs, we compare the performance of our algorithm with the state-of-the-art \emph{non-local} algorithm since, to the best of our knowledge, our local algorithm for digraphs is the first such algorithm in the literature.
All experiments were performed on a Lenovo Yoga 2 Pro with an Intel(R) Core(TM) i7-4510U CPU @ 2.00GHz processor and 8GB of RAM.
 Our code can be downloaded from
\href{https://github.com/pmacg/local-densely-connected-clusters}{https://github.com/pmacg/local-densely-connected-clusters}.





\subsection{Results for Undirected Graphs}
\subsubsection{Experiments on Synthetic Data}
We first compare the performance of our algorithm, \texttt{LocBipartDC}, with the previous state-of-the-art given by  Li \& Peng \cite{LP2013}, which we refer to as \texttt{LP}, on synthetic graphs generated from the stochastic block model (SBM).
Specifically, we assume that the graph has $k=3$ clusters $\{C_j\}_{j=1}^3$, and the number of vertices in each cluster, denoted by $n_1, n_2$ and $n_3$ respectively, satisfy $n_1=n_2=0.1n_3$. Moreover, any pair of vertices $u\in C_i$ and $v\in C_j$ is connected with probability $P_{i,j}$. We   assume that $P_{1,1}=P_{2,2}=p_1$, $P_{3,3}=p_2$, $P_{1,2}=q_1$, and $P_{1,3}=P_{2,3}=q_2$.
 Throughout our experiments, we maintain the ratios $p_2 = 2 p_1$ and $q_2 = 0.1 p_1$, leaving the parameters $n_1$, $p_1$ and $q_1$ free. Notice that the different values of $q_1$ and $q_2$ guarantee that $C_1$ and $C_2$ are the ones optimising  the $\beta$-value, which is why our proposed model  is  slightly more involved than the standard SBM. 




 We evaluate the quality of the output $(L,R)$ returned by each algorithm with respect to its $\beta$-value, the   Adjusted Rand Index (ARI)~\cite{gatesImpactRandomModels2017}, as well as the ratio of the misclassified vertices defined by $\frac{\cardinality{L \triangle C_1} + \cardinality{R \triangle C_2}}{\cardinality{L \union C_1} + \cardinality{R \union C_2}}$, where $A \triangle B$ is the symmetric difference between $A$ and $B$.
All our reported results are the average performance of each algorithm over $10$ runs, in which a random vertex from $C_1\cup C_2$ is chosen as the starting vertex of the algorithm.

\paragraph{Setting the $\epsilon$ parameter in the \texttt{LP} algorithm.}

The \texttt{LP} algorithm has an additional parameter over ours, which we refer to as $\epsilon_{LP}$. This parameter influences the runtime and performance of the algorithm and must be in the range $[0, 1/2]$.
In order to choose a fair value of $\epsilon_{LP}$ for comparison with our algorithm, we consider several values on graphs with a range of target volumes.

We generate graphs from the SBM such that $p_1 = 1 / n_1$ and $q_1 = 100 p_1$ and vary the size of the target set by varying $n_1$ between $100$ and $10,000$.
For values of $\epsilon_{LP}$ in $\{0.01, 0.02, 0.1, 0.5\}$, Figure~\ref{fig:choosing_eps}(a) shows how the runtime of the \texttt{LP} algorithm compares to the \texttt{LocBipartDC} algorithm for a range of target volumes.
In this experiment, the runtime of the \texttt{LocBipartDC} algorithm lies between the runtimes of the \texttt{LP} algorithm for $\epsilon_{LP} = 0.01$ and $\epsilon_{LP} = 0.02$.
However, Figure~\ref{fig:choosing_eps}(b) shows that the performance of the \texttt{LP} algorithm with $\epsilon_{LP} = 0.01$ is significantly worse than the performance of the \texttt{LocBipartDC} algorithm and so for a fair comparison, we set $\epsilon_{LP} = 0.02$ for the remainder of the experiments.

\begin{figure}[h]
\centering
\subfigure[]{\includegraphics[width=0.3\columnwidth]{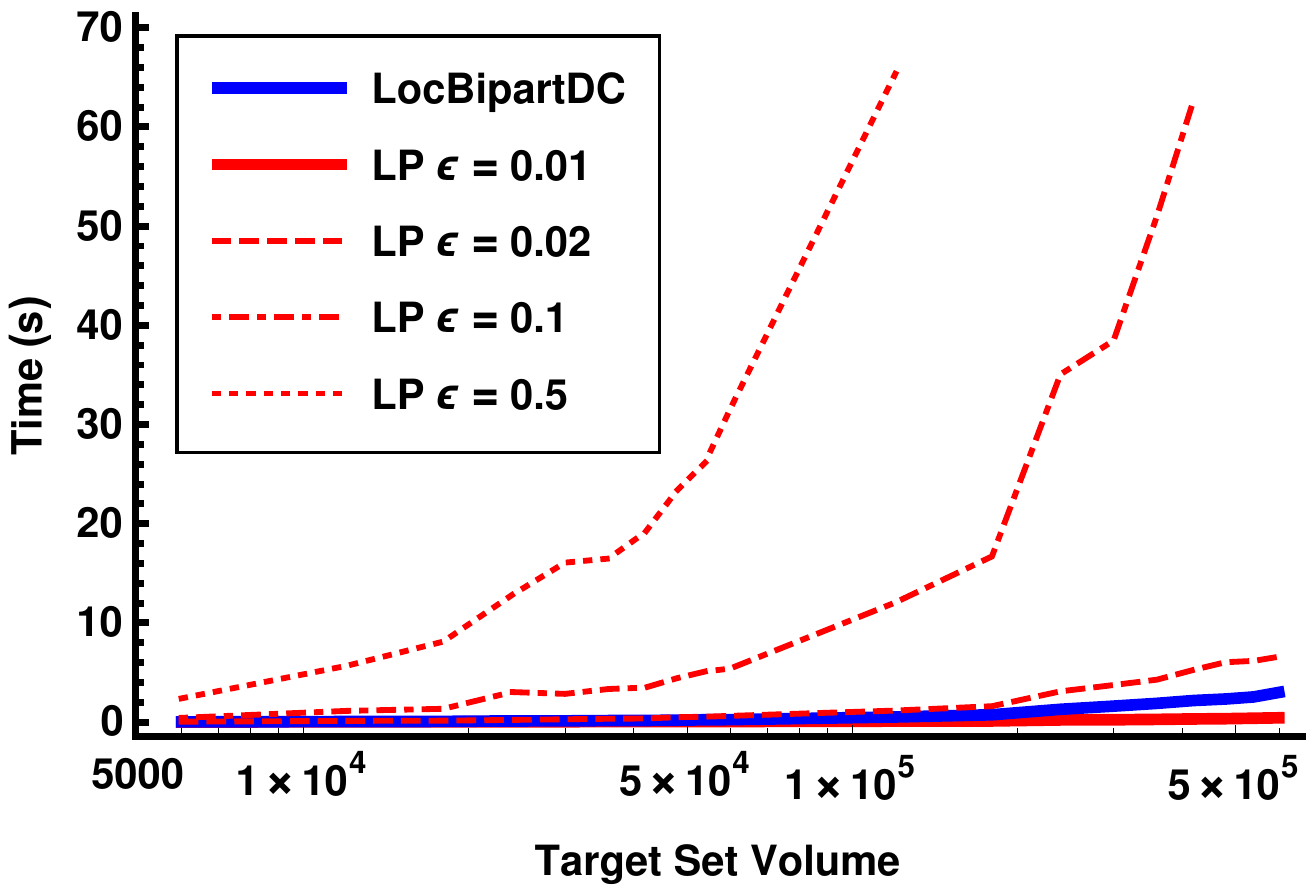}}
\hspace{1in}
\subfigure[]{\includegraphics[width=0.3\columnwidth]{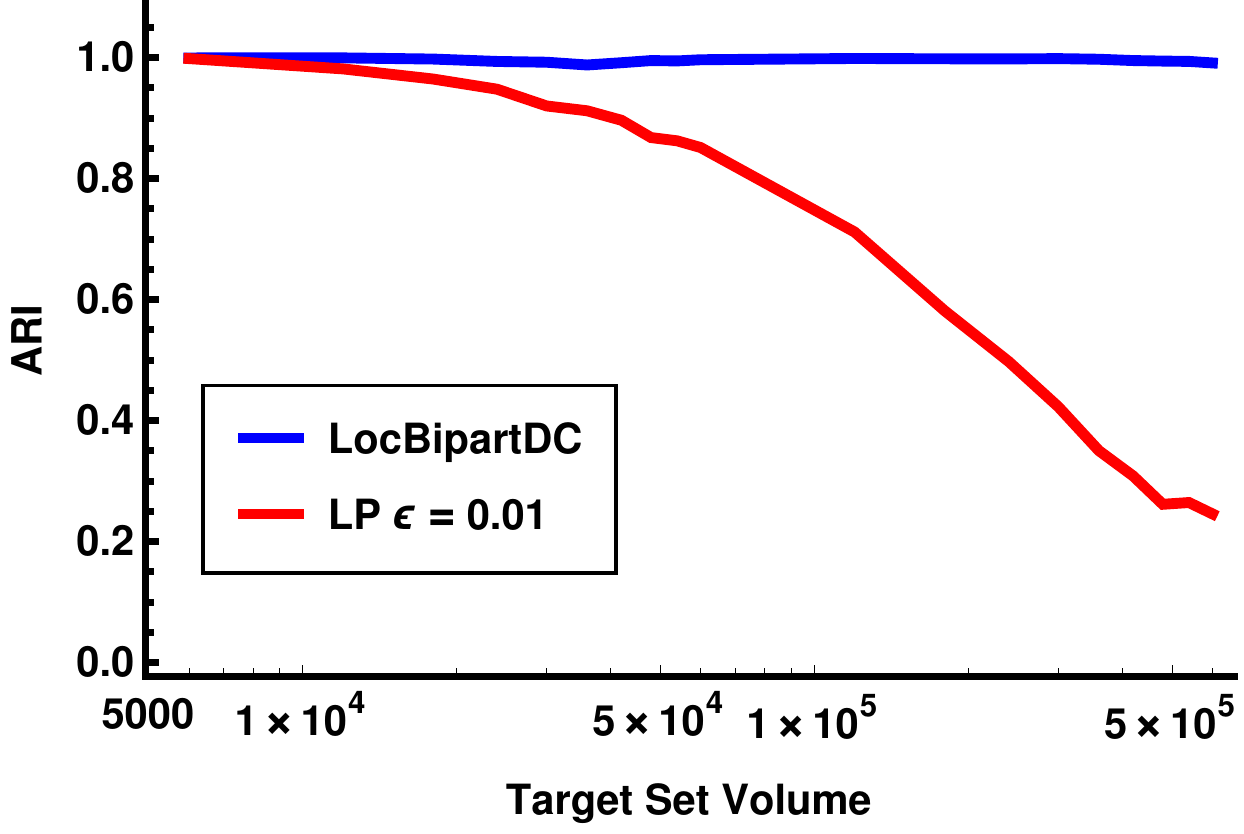}}
\caption{\small{(a) A comparison of the runtimes of the \texttt{LocBipartDC} algorithm and the \texttt{LP} algorithm with various values of $\epsilon_{LP}$.
(b) Setting $\epsilon_{LP} = 0.01$ results in significantly worse performance than the $\texttt{LocBipartDC}$ algorithm.
\label{fig:choosing_eps}}}
\end{figure}

\paragraph{Comparison on the Synthetic Data.}
\begin{table*}[t]
\caption{\small Full comparison between    Algorithm~\ref{alg:local_max_cut}~(\texttt{LBDC}) and the previous state-of-the-art~\cite{LP2013} (\texttt{LP}). For clarity we report the target bipartiteness $\beta = \beta(C_1, C_2)$ and target volume $\gamma = \vol(C_1 \union C_2)$ along with the SBM parameters.}
\label{tab:full_results}
\vskip 0.1in
\centering
\begin{small}
\begin{sc}
\begin{tabular}{cccccc}
\toprule
Input graph parameters & Algo. & Runtime & $\beta$-value & ARI & Misclassified Ratio\\
\midrule
$n_1 = 1,000$, $p_1 = 0.001$, $q_1 = 0.018$ & \texttt{LBDC} & \textbf{0.09} & \textbf{0.154} & \textbf{0.968} & \textbf{0.073} \\
$\beta \approx 0.1$, $\gamma \approx 40,000$ & \texttt{LP} & 0.146 & 0.202 & 0.909 & 0.138 \\
\midrule
$n_1 = 10,000$, $p_1 = 0.0001$, $q_1 = 0.0018$ & \texttt{LBDC} & \textbf{0.992} & \textbf{0.215} & \textbf{0.940} & \textbf{0.145} \\
$\beta \approx 0.1$, $\gamma \approx 400,000$ & \texttt{LP} & 1.327 & 0.297 & 0.857 & 0.256 \\
\midrule
$n_1 = 100,000$, $p_1 = 0.00001$, $q_1 = 0.00018$ & \texttt{LBDC} & \textbf{19.585} & \textbf{0.250} & \textbf{0.950} & \textbf{0.166} \\
$\beta \approx 0.1$, $\gamma \approx 4,000,000$ & \texttt{LP} & 30.285 & 0.300 & 0.865 & 0.225 \\
\midrule
$n_1 = 1,000$, $p_1 = 0.004$, $q_1 = 0.012$ & \texttt{LBDC} & \textbf{1.249} & \textbf{0.506} & \textbf{0.503} & \textbf{0.763} \\
$\beta \approx 0.4$, $\gamma \approx 40,000$ & \texttt{LP} & 1.329 & 0.597 & 0.445 & 0.785 \\
\bottomrule
\end{tabular}
\end{sc}
\end{small}
\vskip -0.1in
\end{table*}
We now compare the \texttt{LocBipartDC} and \texttt{LP} algorithms'
performance on graphs generated from the SBM with different values of $n_1,p_1$ and $q_1$.
As shown in Table~\ref{tab:full_results}, our algorithm not only runs faster, but also produces better clusters with respect to all three metrics.
Secondly, since the clustering task becomes more challenging when the target clusters have higher $\beta$-value,
we compare the algorithms' performance on a sequence of instances with increasing value of $\beta$.
 Since $q_1/p_1 = 2(1-\beta)/\beta$, we simply fix the values of  $n_1, p_1$ as $n_1=1,000,p_1=0.001$, and generate graphs with increasing value of $q_1/p_1$; this gives us graphs with monotone values of $\beta$.
 As shown in Figure~\ref{fig:bipvsari}, our algorithms performance is always better than the previous stat-of-the-art.

\begin{figure}[htpb]
\centering
    \subfigure[] {\label{fig:bipvsari}\includegraphics[width=0.3\columnwidth]{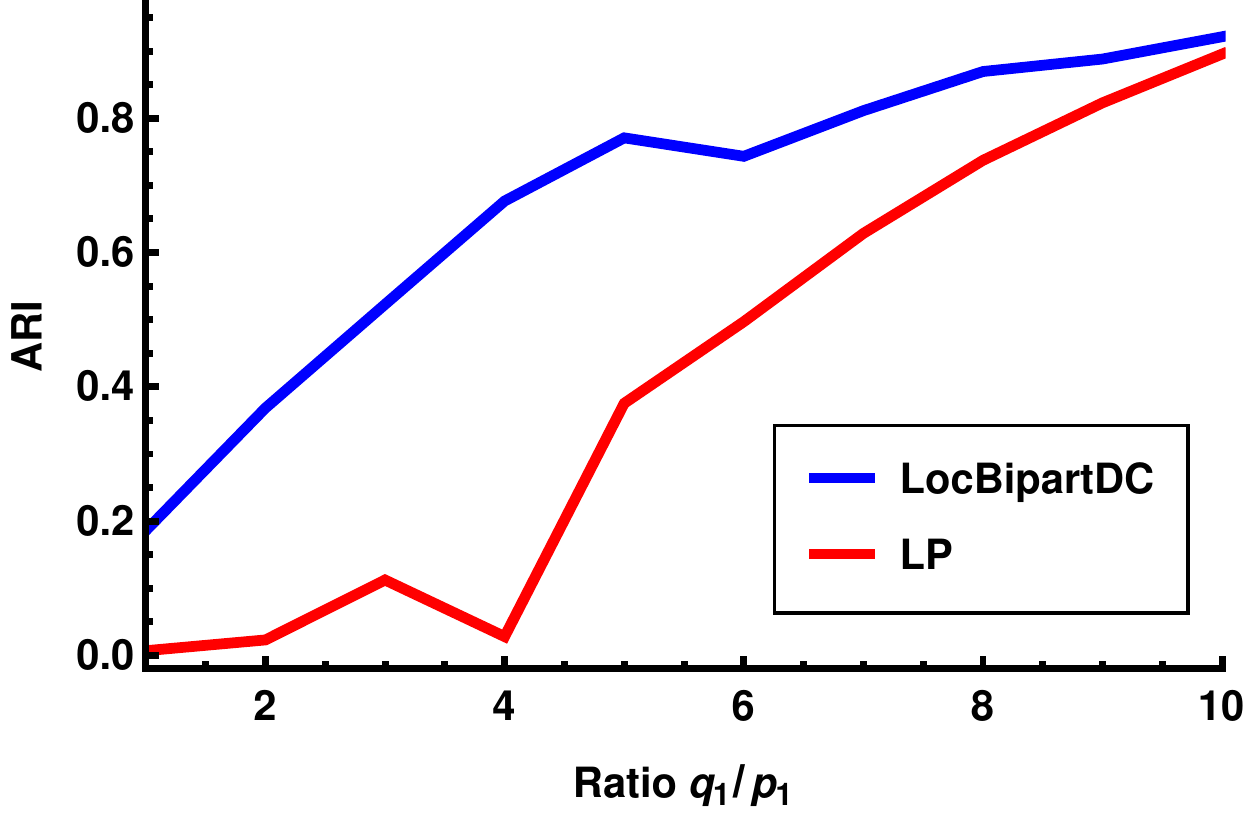}}
    \hspace{1in}
    \subfigure[] 
    {\label{fig:timevsari}\includegraphics[width=0.3\columnwidth]{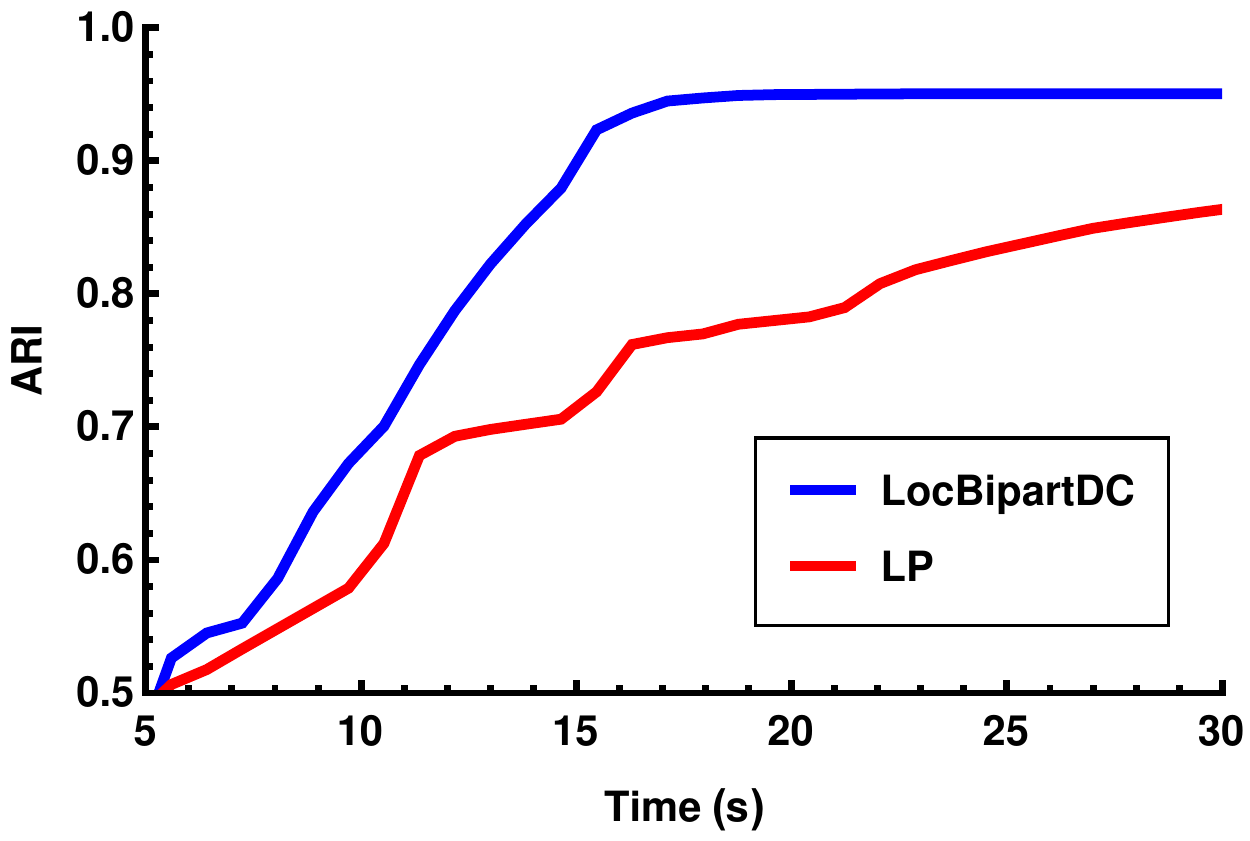}}
    \caption{
    \small
         (a) The ARI of each algorithm when varying the target $\beta$-value. A larger $q_1 / p_1$ ratio corresponds to a smaller $\beta$-value. (b) The ARI of each algorithm when bounding the runtime of the algorithm.
    \label{fig:ari}
    }
\end{figure}


 Thirdly, notice that both algorithms use some parameters to control the algorithm's runtime and the output's approximation ratio, which are naturally influenced by each other. To study this dependency, we generate graphs according to   $n_1 = 100,000$,  $p_1 = 0.000015$, and $q_1 = 0.00027$ which results in target sets with     $\beta \approx 0.1$ and volume $\gamma \approx 6,000,000$.
 Figure~\ref{fig:ari}(b) shows that, in comparison with the previous state-of-the-art, our algorithm  takes much less time to produce output with the same ARI value.

\subsubsection{Experiments on the Real-world Dataset} 
We demonstrate the significance of our algorithm on the Dyadic Militarised Interstate Disputes Dataset~(v3.1)~\cite{midDataset}, which
records every 
interstate dispute during 1816--2010, including the level of hostility resulting from the dispute and the number of casualties, and has been widely studied in the social and political sciences~\cite{mansfieldWhyDemocraciesCooperate2002, martinMakeTradeNot2008} as well as the machine learning    community~\cite{huDeepGenerativeModels2017, menonLinkPredictionMatrix2011, traagCommunityDetectionNetworks2009}.
For a given time period, we construct a graph from the data by representing each country with a vertex and adding an edge between each pair of countries weighted according to the severity of any military disputes between those countries.
Specifically, if there's a war\footnote{A war is defined by the maintainers of the dataset as a series of battles resulting  in at least 1,000 deaths.
}
between the two countries, the corresponding two vertices are 
 connected by an edge with weight $30$; for any other dispute that is not part of an interstate war, the two corresponding vertices are connected by an edge with weight $1$.
 We always use the USA as the starting vertex of the algorithm, and   our algorithm's output, as visualised in Figure~\ref{fig:intro_examples}(a)-(d), can be well explained by geopolitics. The 
 $\beta$-values of the pairs of clusters in Figures~\ref{fig:intro_examples}(a)-(d) are   $0.361$, $0.356$, $0.170$ and $0.191$ respectively.
 
\subsection{Results for Digraphs}
Next we evaluate the performance of our algorithm for digraphs on synthetic and real-world datasets.
Since there are no previous local digraph clustering algorithms that achieve similar objectives to ours, we compare the output of Algorithm~\ref{alg:directed_evo_cut} (\texttt{ECD}) with the state-of-the-art
non-local algorithm proposed by  Cucuringu et al.~\cite{CLS+2020},  and we refer this to as \texttt{CLSZ} in the following. 
 
\paragraph{Synthetic Dataset.}
 
We first look at the  \emph{cyclic block model}~(CBM) described in  Cucuringu et al.~\cite{CLS+2020} with parameters $k$, $n$, $p$, $q$, and $\eta$.
In this model, we generate a digraph with $k$ clusters $C_1, \ldots, C_k$ of size $n$, and for $u, v \in C_i$, there is an edge between $u$ and $v$ with probability $p$  and the edge direction is chosen uniformly at random. For $u \in C_i$ and $v \in C_{i + 1 \mod k}$, there is an edge between $u$ and $v$ with probability $q$, and the edge is directed from $u$ to $v$ with probability $\eta$ and from $v$ to $u$ with probability $1 - \eta$.
Throughout our experiments, we fix $p = 0.001$, $q = 0.01$, and $\eta = 0.9$.

Secondly,  since the goal of our algorithm is to find local structure in a graph, we extend the cyclic block model with additional local clusters and refer to this model as CBM+.
In addition to the parameters of the CBM, we introduce the parameters $n', q'_1, q'_2$, and $\eta'$.
In this model, the clusters $C_1$ to $C_k$ are generated as in the CBM, and there are two additional clusters $C_{k+1}$ and $C_{k+2}$ of size $n'$.
For $u, v \in C_{k+i}$ for $i \in \{1, 2\}$, there is an edge between $u$ and $v$ with probability $p$
and for $u \in C_{k+1}$ and $v \in C_{k+2}$, there is an edge with probability $q'_1$;
 the edge directions are chosen uniformly at random.
For $u \in C_{k+1} \union C_{k+2}$ and $v \in C_1$, there is an edge with probability $q'_2$.
If $u \in C_{k+1}$, the orientation is from $v$ to $u$ with probability $\eta'$ and from $u$ to $v$ with probability $1 - \eta'$ and if $u \in C_{k+2}$, the orientation is from $u$ to $v$ with probability $\eta'$ and from $v$ to $u$ with probability $1 - \eta'$.
We always fix $q'_1 = 0.5$, $q'_2 = 0.005$, and $\eta' = 1$.
Notice that the clusters $C_{k+1}$ and $C_{k+2}$ form a ``local'' cycle with the cluster $C_1$.

In Table~\ref{tab:gdsbm}, we report the average performance over 10 runs with a variety of parameters.
We find that \texttt{CLSZ} can uncover the global structure in the CBM more accurately than \texttt{ECD}.
On the other hand, \texttt{CLSZ} fails to identify the local cycle in the CBM+ model, while \texttt{ECD} succeeds. 

\begin{table}[htpb]
\caption{\small Comparison of \texttt{ECD} with \texttt{CLSZ} on synthetic data. \label{tab:gdsbm}}
\begin{center}
\begin{small}
\begin{sc}
\small\addtolength{\tabcolsep}{-1pt}
\begin{tabular}{cccccccc}
\toprule
& & & & \multicolumn{2}{c}{Time} & \multicolumn{2}{c}{ARI} \\
\cmidrule(lr){5-6}\cmidrule(lr){7-8}
Model & $n'$ & $n$ & $k$ & \texttt{ECD} & \texttt{CLSZ} & \texttt{ECD} & \texttt{CLSZ} \\
\midrule
CBM & - & $10^3$ & $5$ & $\mathbf{1.59}$ & $3.99$ & $0.92$ & $\mathbf{1.00}$ \\
CBM & - &$10^3$ & $50$ & $\mathbf{3.81}$ & $156.24$ & $0.99$ & $0.99$ \\
CBM+ & $10^2$ & $10^3$ & $3$ & $\mathbf{0.24}$ & $6.12$ & $\mathbf{0.98}$ & $0.35$ \\
CBM+ & $10^2$ & $10^4$ & $3$ & $\mathbf{0.32}$ & $45.17$ & $\mathbf{0.99}$ & $0.01$ \\
\bottomrule
\end{tabular}
\end{sc}
\end{small}
\end{center}
\vskip -0.1in
\end{table}

\paragraph{Real-world Dataset.}

Now  we evaluate the algorithms' performance
on the US Migration Dataset~\cite{census2000}.
For fair comparison, we follow  Cucuringu et al.~\cite{CLS+2020} and construct the digraph as follows: every county in the mainland USA is represented by a vertex; for any vertices $j,\ell$, the edge weight of $(j,\ell)$ is given by $\left|\frac{M_{j,\ell} - M_{\ell,j}}{M_{j,\ell} + M_{\ell,j}}\right|$, where $M_{j,\ell}$ is the number of people who migrated from county $j$ to county $\ell$ between 1995 and 2000; in addition, the direction of $(j,\ell)$ is set to be from $j$ to $\ell$ if $M_{j,\ell}>M_{\ell,j}$, otherwise the direction is set to be the opposite.

For \texttt{CLSZ}, we follow their suggestion on the same dataset and set $k=10$.
Both algorithms' performance is evaluated with respect to the flow ratio, as well as the Cut Imbalance ratio used in their work.
For any vertex sets $L$ and $R$, the cut imbalance ratio is defined by $\mathrm{CI}(L,R) = \frac{1}{2}\cdot \left|\frac{e(L,R) - e(R,L)}{ \ e(L,R) + e(R,L)}\right|$, and a higher $\mathrm{CI}(L,R)$ value indicates the connection between $L$ and $R$ is more significant.  Using counties in Ohio, New York, California, and Florida as the starting vertices, our algorithm's outputs are  visualised in Figures~\ref{fig:intro_examples}(e)-(h),
and in Table~\ref{tab:compalgo2} we compare them
to  the top $4$ pairs returned by \texttt{CLSZ}, which are shown in Figure~\ref{fig:migration}.
Our algorithm produces better outputs with respect to both metrics.

\begin{table*}[thpb]
\caption{Comparison of \texttt{EvoCutDirected} with \texttt{CLSZ} on the US migration dataset.}
\label{tab:compalgo2}
\vskip 0.15in
\begin{center}
\begin{small}
\begin{sc}
\begin{tabular}{ccccc}
\toprule
Figure & Algorithm      & Cluster         & Cut Imbalance & Flow Ratio \\
\midrule
\ref{fig:migration}(a)       & CLSZ           & Pair 1          & $0.41$        & $0.80$     \\
\ref{fig:migration}(b)       & CLSZ           & Pair 2          & $0.35$        & $0.83$     \\
\ref{fig:migration}(c)       & CLSZ           & Pair 3          & $0.32$        & $0.84$     \\
\ref{fig:migration}(d)       & CLSZ           & Pair 4          & $0.29$        & $0.84$     \\
\ref{fig:intro_examples}(e)       & EvoCutDirected & Ohio Seed       & $0.50$        & $0.56$     \\
\ref{fig:intro_examples}(f)       & EvoCutDirected & New York Seed   & $0.49$        & $0.58$     \\
\ref{fig:intro_examples}(g)       & EvoCutDirected & California Seed & $0.49$        & $0.67$     \\
\ref{fig:intro_examples}(h)       & EvoCutDirected & Florida Seed    & $0.42$        & $0.79$     \\
\bottomrule
\end{tabular}
\end{sc}
\end{small}
\end{center}
\vskip -0.1in
\end{table*}

\begin{figure*}[htpb]
\centering
    \subfigure[\texttt{CLSZ}, Pair 1] {\includegraphics[width=0.24\columnwidth]{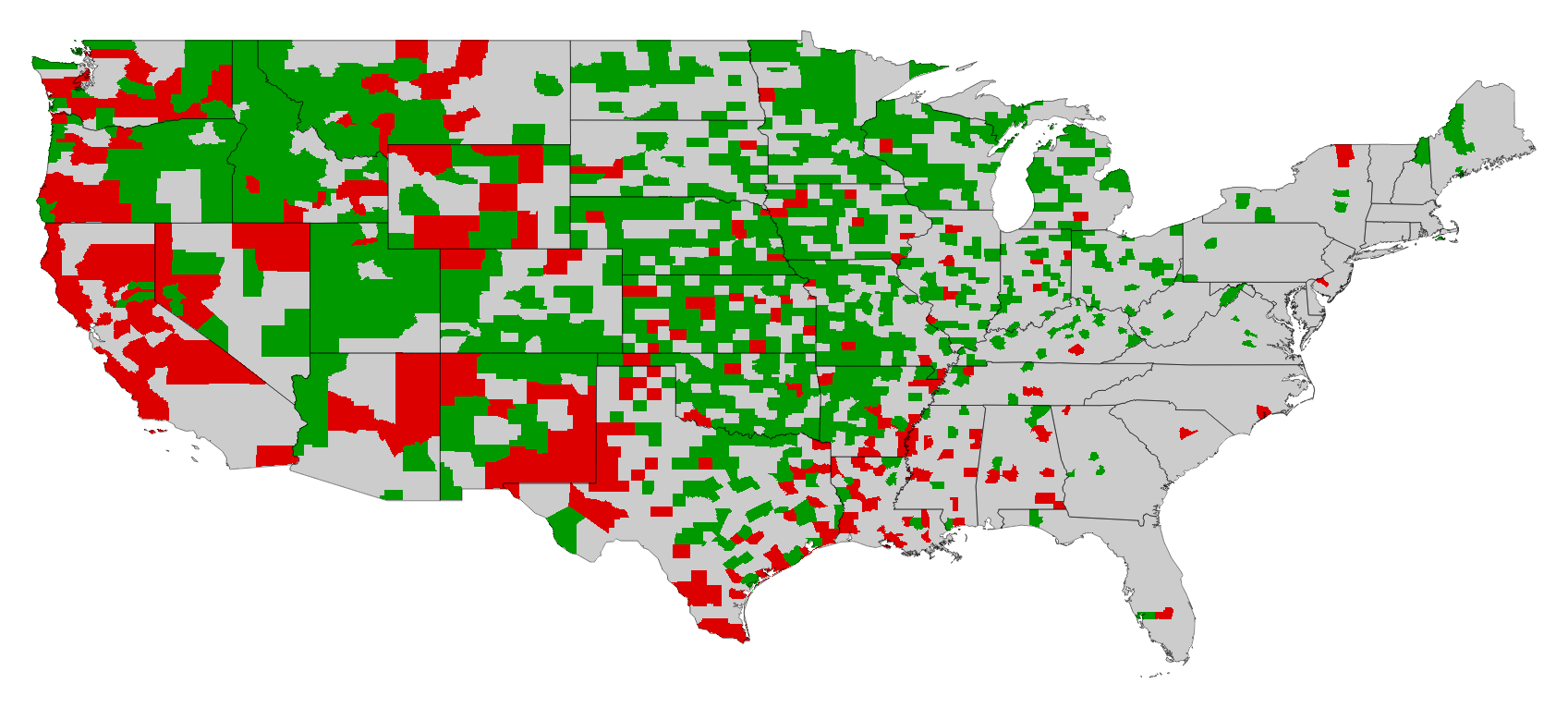}}
    \subfigure[\texttt{CLSZ}, Pair 2] {\includegraphics[width=0.24\columnwidth]{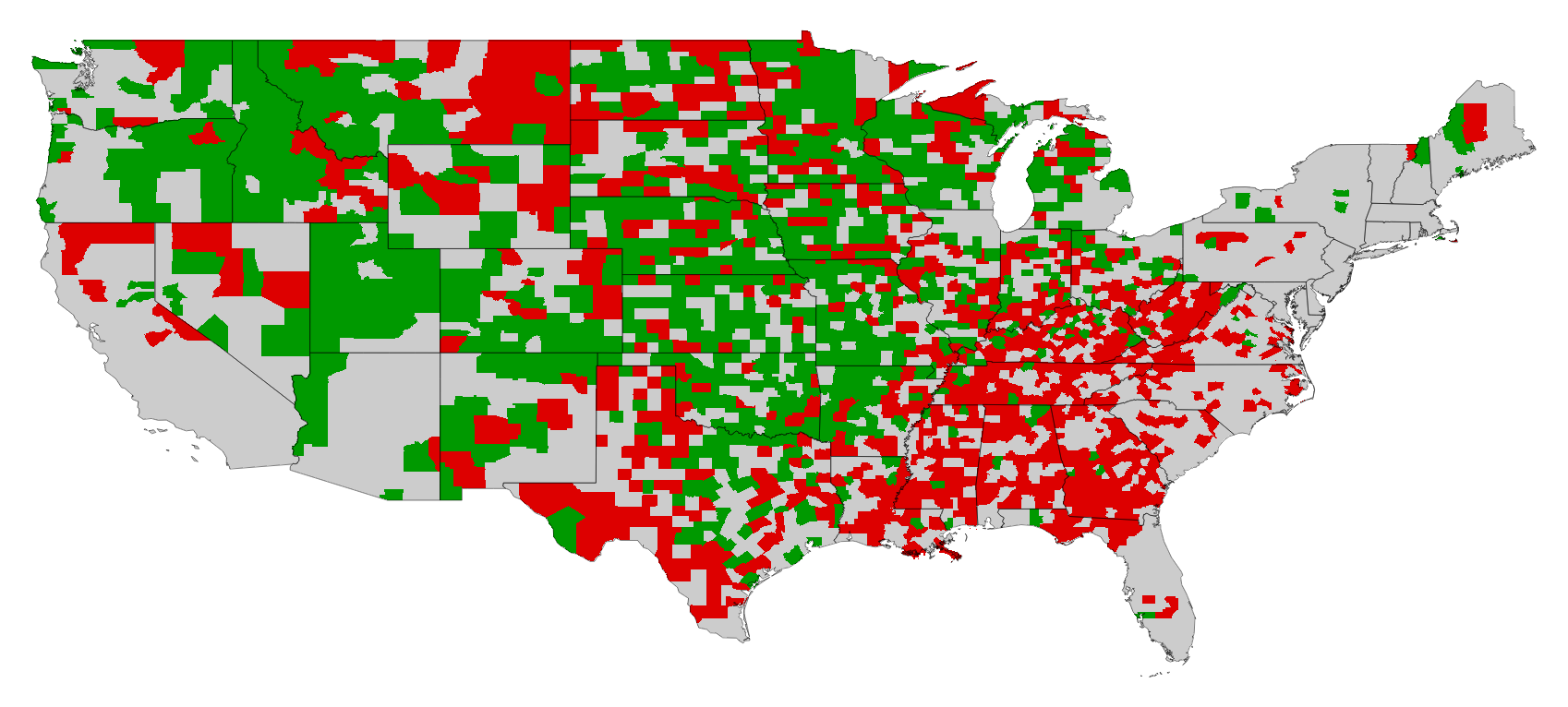}}
    \subfigure[\texttt{CLSZ}, Pair 3] {\includegraphics[width=0.24\columnwidth]{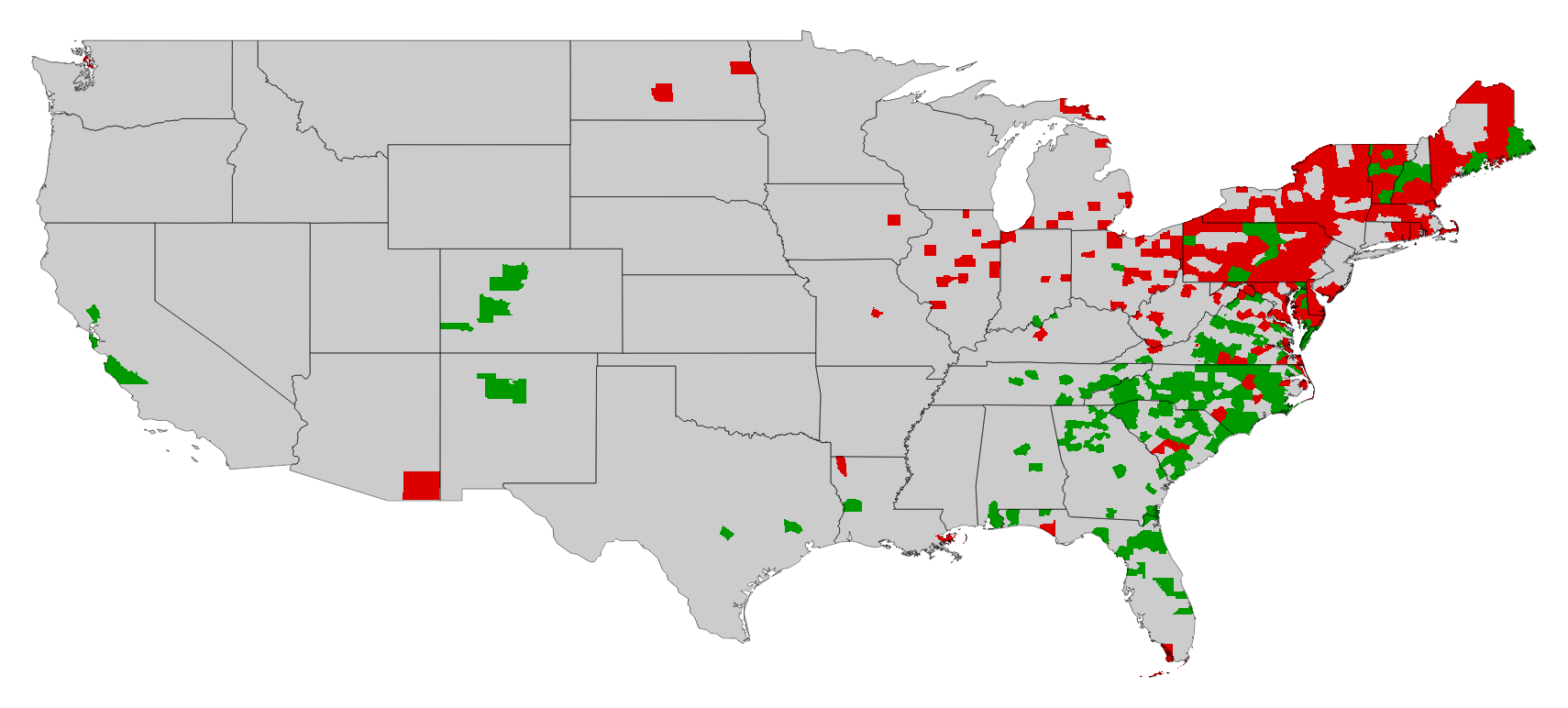}}
    \subfigure[\texttt{CLSZ}, Pair 4] {\includegraphics[width=0.24\columnwidth]{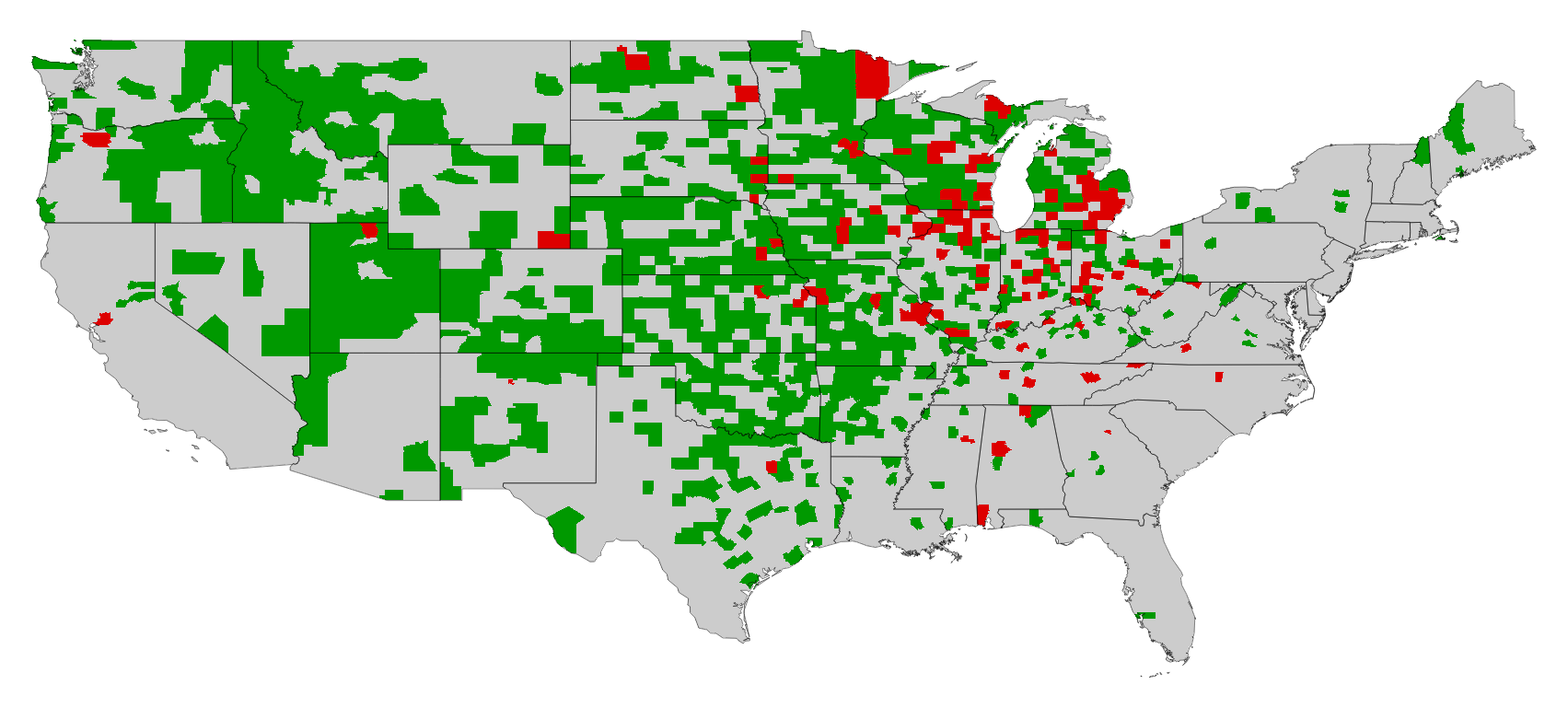}}
    \caption{\small{
    The top $4$ pairs of clusters found by   \texttt{CLSZ}   on the US Migration dataset.
    \label{fig:migration}
    }}
\end{figure*}

 These experiments suggest that local algorithms are not only more efficient, but can also be much more effective than global algorithms when learning certain  structures  in graphs. In particular, some localised structure
might be hidden when applying the objective function over the entire graph.

\section*{Acknowledgements}
We thank Luca Zanetti for making insightful comments on an earlier version of this work.
Peter Macgregor is supported by the Langmuir PhD Scholarship, and
He Sun is supported by an EPSRC Early Career Fellowship (EP/T00729X/1).

\bibliography{references.bib, datasets.bib}
\bibliographystyle{plain}

\appendix


\section{Omitted detail from Section \ref{sec:undirected}}


 In this section, we present all the omitted details from Section~\ref{sec:undirected} of our paper.
We formally introduce the Lov\'asz-Simonovits curve and use it to analyse the performance of the \texttt{LocBipartDC} algorithm for undirected graphs.

\paragraph{The Lov\'asz-Simonovits Curve.}
Our analysis of the \texttt{LocBipartDC} algorithm is based on the Lov\'asz-Simonovits curve, which has been used extensively in the analysis of local clustering algorithms\ \cite{ACL2006, ZLM2013}.
Lov\'asz and Simonovits~\cite{LS1990} originally defined the following function to reason  about the mixing rate of Markov chains. 
For any vector $p \in \mathbb{R}_{\geq 0}^{n}$, we  order the vertices such that $$\frac{p(v_1)}{\deg(v_1)} \geq \frac{p(v_2)}{\deg(v_2)} \geq \ldots \geq \frac{p(v_n)}{\deg(v_n)},$$ 
and define the sweep sets  of $p$  as
$
    S_j^p = \{v_1, \ldots, v_j\}
$ for $1\leq j\leq n$.
The Lov\'asz-Simonovits curve, denoted by $p[x]$ for $x \in [0, \vol(V)]$, is defined   by the points
$
    p[\vol(S_j^p)] \triangleq p(S_j^p ),
$
and is linear between those points for consecutive $j$.  
Lov\'asz and Simonovits also show that
\begin{equation} \label{eq:deflscurve}
    \lscurve{x} = \max_{\substack{w \in [0, 1]^{n} \\ \sum_{u\in V } w(u) \deg(u) = x}} \sum_{u \in V} w(u) p(u).
\end{equation}

\begin{proposition} \label{prop:lsleq}
    For any $p \in \R_{\geq 0}^{n}$ and any $S \subseteq V$, it holds that 
   $
        p(S) \leq \lscurve{\vol(S)}.
    $
\end{proposition}
\begin{proof} 
    By definition, we have that
 $   p(S)  =  \sum_{u \in V} \chi_{S}(u) p(u).
$
  Hence, by   (\ref{eq:deflscurve}) we have  that
    $        p(S) \leq \lscurve{\vol(S)}$.
\end{proof}

\paragraph{Double Cover Reduction.}
The following two proofs establish the reduction from almost-bipartite sets to sets with low conductance in the double cover and follow almost directly from the definition of the bipartiteness ratio and conductance.
\begin{proof}[Proof of Lemma~\ref{lem:cond_bip}]
    Let $S' = L_1 \cup R_2$, and by definition we have  that $
        \vol_{G}(S) = \vol_{H}(S')
    $. On the other hand, it holds that 
    \[
    \vol_H(S') = e_H(S', V\setminus S') + 2e_H(S', S') = e_H(S', V \setminus S') + 2e_H(L_1, R_2),
    \]
    which implies that 
    \[
    \Phi_{H}(S') = \frac{e_H(S', V\setminus S')}{\vol_H(S')} = \frac{\vol_G(S) - 2e_H(L_1, R_2)}{\vol_G(S)} = 1 - \frac{2e_G(L,R)}{\vol_G(S)} = \beta_G(L,R),
    \]
    which proves the statement.
\end{proof}

\begin{proof}[Proof of Lemma~\ref{lem:ReductionForSimpleset}]
    The  statement follows by  Lemma~\ref{lem:cond_bip}.
\end{proof}

\paragraph{The simplify operator.} We now study the $\sigma$-operator to establish the properties given in Lemma~\ref{lem:sigmaproperty}.
Recall that the lazy random walk matrix for a graph is defined to be $W = \frac{1}{2}(I + D^{-1} A)$.
\begin{proof}[Proof of Lemma~\ref{lem:sigmaproperty}]
For any $u_1 \in V_H$, we have by definition that 
\begin{align*}
    \sigma\circ(c\cdot  p)(u_1) & = \max(0, c\cdot p(u_1) - c\cdot  p(u_2)) \\
    & = c\cdot \max(0, p(u_1) - p(u_2)) \\
    & = c \cdot \sigma\circ p(u_1),
\end{align*}
and similarly we have $\sigma\circ(c\cdot  p)(u_2) = c \cdot \sigma\circ p(u_2)$. Therefore, the first property holds.
    
Now we prove the second statement. Let $u_1, u_2 \in V_H$ be the vertices  that correspond to $u\in V_G$. We have by definition that 
\begin{align*}
    \sigma\circ(a+b)(u_1) & = \max (0, a(u_1) + b(u_1) - a(u_2) - b(u_2) )  \\ 
   &  = \max (0, a(u_1)  - a(u_2) + b(u_1) - b(u_2) ) \\
   & \leq \max (0, a(u_1)  - a(u_2)) + \max (0, b(u_1) - b(u_2) )\\
   &  = \sigma\circ a (u_1 ) + \sigma\circ b (u_1),
\end{align*}
where the last inequality holds by the fact that $\max(0, x+y) \leq \max (0, x) + \max(0,y) $ for any $x,y\in\mathbb{R}$.
For the same reason, we have that 
\[
    \sigma\circ(a+b) (u_2) \leq \sigma\circ a (u_2 ) + \sigma\circ b (u_2).
\]
Combining these proves the second property of $\sigma$.
    
Finally, we will prove the third property.
By the definition of matrix $W$, we have that
\[
(p W)(u_1) = \frac{1}{2}p(u_1) + \frac{1}{2} \sum_{v \in N_G(u)} \frac{p(v_2)}{\deg_H(v_2)},
\]
and
\[
(p W)(u_2) = \frac{1}{2}p(u_2) + \frac{1}{2} \sum_{v \in N_G(u)} \frac{p(v_1)}{\deg_H(v_1)}.
\]
Without loss of generality, we assume that $(p W)(u_1) \geq (p W)(u_2)$. This implies that $(\sigma\circ(pW))(u_2) = 0 \leq (\sigmap W)(u_2)$, hence it suffices to prove that $(\sigma\circ(pW))(u_1) \leq (\sigmap W)(u_1)$. By definition, we have that 
\begin{align}
(\sigma\circ(pW))(u_1)   &  = \max(0, pW(u_1) -pW(u_2))   \notag \\
& = pW(u_1) -pW(u_2) \notag\\
& = \frac{1}{2}\left( p(u_1) -p(u_2)\right)  + \frac{1}{2} \sum_{v \in N_G(u)} \left( \frac{p(v_2)}{\deg_H(v_2)} -\frac{p(v_1)}{\deg_H(v_1)} \right) \notag\\
& = \frac{1}{2}\left( p(u_1) -p(u_2)\right)  + \frac{1}{2} \sum_{v \in N_G(u)} \frac{ p(v_2)  - p(v_1)}{\deg_G(v)}, \label{eq:sigmadif}
\end{align}
where the last line follows by the fact that $\deg_H(v_1) = \deg_H(v_2) = \deg_G(v)$  for any $v\in V_G$.  To analyse   \eqref{eq:sigmadif}, notice that $p(v_1 ) - p(v_2) =(\sigma\circ p)(v_1) -\sigmap (v_2)$ for any $v\in V_G$ for the following reasons:
\begin{itemize}
    \item When $p(v_1)\geq p(v_2)$, we have that $
(\sigma\circ p)(v_1) -\sigmap (v_2) = \left(p(v_1) - p(v_2)\right) - 0 = p(v_1) - p(v_2)$;
    \item Otherwise, we have $p(v_1)<p(v_2)$ and $
(\sigma\circ p)(v_1) -\sigmap (v_2) = 0 - (p(v_2) - p(v_1)) = p(v_1) - p(v_2)$.
\end{itemize} 
Therefore, we can rewrite \eqref{eq:sigmadif}
 as  
\begin{align*}
        (\sigma\circ (pW))(u_1)  
       &  = \frac{1}{2}((\sigma\circ p)(u_1) - (\sigma\circ p)(u_2)) + \frac{1}{2} \sum_{v \in N(u)} \frac{ \sigmap(v_2) - \sigmap(v_1)}{\deg_G(v)}  \\
   & \leq \frac{1}{2}\cdot\sigmap(u_1) + \frac{1}{2}\sum_{v \in N(u)} \frac{ \sigmap(v_2) }{\deg_G(v)} \\
    & = (\sigmap W)(u_1),
    \end{align*}
where the inequality holds by the fact that $\sigmap(u_1)\geq 0$ and $\sigmap(u_2)\geq 0$ for any $u\in V_G$.
Since the analysis above holds for any $v\in V_G$, we have that $\sigma\circ(pW)\preceq \sigmap W$, which finishes the proof of the statement.
\end{proof}

\paragraph{Analysis of $\texttt{ApproximatePagerankDC}$.}
The following lemma shows that \texttt{dcpush} maintains that throughout the algorithm, $\left[p_1, p_2\right]$ is an approximate Pagerank vector on the double cover with residual $\left[r_1, r_2\right]$.
\begin{lemma} \label{lem:dcpush}
Let $p'_1, p'_2, r'_1$ and $r'_2$ be the result of $\texttt{dcpush}((u, i), p_1, p_2, r_1, r_2)$. Let $p = \left[p_1, p_2 \right]$, $r = \left[r_1, r_2 \right]$, $p' = \left[p'_1, p'_2 \right]$ and $r' = \left[r'_1, r'_2 \right]$.
Then, we have that
$$
p' + \pr_H(\alpha, r') = p + \pr_H(\alpha, r).$$
\end{lemma}

\begin{proof}
    After the push operation, we have
    \begin{align*}
        p' & = p + \alpha r(u_i) \chi_{u_i} \\
        r' & = r - r(u_i) \chi_{u_i} + (1 - \alpha) r(u_i) \chi_{u_i} W
    \end{align*}
    where $W$ is the lazy random walk matrix on the double cover of the input graph.
    Then we have
    \begin{align*}
        p + \pr_H(\alpha, r) & = p + \pr_H(\alpha, r - r(u_i) \chi_{u_i}) + \pr_H(\alpha, r(u_i) \chi_{u_i}) \\
        & = p + \pr_H(\alpha, r - r(u_i) \chi_{u_i}) + \left[\alpha r(u) \chi_{u_i} + (1 - \alpha) \pr_H(\alpha, r(u) \chi_{u_i} W) \right] \\
        & = \left[ p + \alpha r(u_i) \chi_{u_i}\right] + \pr_H(\alpha, \left[r - r(u_i) \chi_{u_i} + (1 - \alpha)r(u_i)\chi_{u_i}W\right]) \\
        & = p' + \pr_H(\alpha, r'),
    \end{align*}
    where we use the fact that $\pr(\alpha, \cdot)$ is linear and that
    \[
        \pr(\alpha, s) = \alpha s + (1 - \alpha) \pr(\alpha, sW),
    \]
    which follows from the fact that $\pr(\alpha, s) = \alpha \sum_{t = 0}^\infty (1 - \alpha)^t s W^t$.
\end{proof}

The following lemma shows the correctness and performance guarantee of the $\texttt{ApproximatePagerankDC}$ algorithm.
\begin{lemma} \label{lem:aprdc}
    For a graph $G$ with double cover $H$ and any $v \in V_G$, $\texttt{ApproximatePagerankDC}(v, \alpha, \epsilon)$ runs in time $O(\frac{1}{\epsilon \alpha})$ and computes $p = \apr_H(\alpha, \chi_{v_1}, r)$ such that $\max_{u \in V_H} \frac{r(u)}{\deg(u)} < \epsilon$ and $\vol(\mathrm{supp}(p))) \leq \frac{1}{\epsilon \alpha}$.
\end{lemma}

\begin{proof}
    By Lemma~\ref{lem:dcpush}, we have that $p = \apr_H(\alpha, \chi_{v_1}, r)$ throughout the algorithm and so this clearly holds for the output vector.
    By the stopping condition of the algorithm, it is clear that $\max_{u \in V_H} \frac{r(u)}{\mathrm{deg}(u)} < \epsilon$.
    
    Let $T$ be the total number of push operations performed by $\texttt{ApproximatePagerankDC}$, and let $d_j$ be the degree of the vertex $u_i$ used in the $j$th push.
    The amount of probability mass on $u_i$ at time $j$ is at least $\epsilon d_j$ and so the amount of probability mass that moves from $r_i$ to $p_i$ at step $j$ is at least $\alpha \epsilon d_i$.
    Since $\|r_1\|_1 + \|r_2\|_1 = 1$ at the beginning of the algorithm, we have that 
    $
        \alpha \epsilon \sum_{j = 1}^T d_j  \leq 1,   
    $
    which implies that 
    \begin{align}
        \sum_{j = 1}^T d_j & \leq \frac{1}{\alpha \epsilon}. \label{eq:pushnum}
    \end{align}
    Now, note that for every vertex in $\mathrm{supp}(p)$, there must have been at least one push operation on that vertex. So,
    \[
        \vol(\mathrm{supp}(p)) = \sum_{v \in \mathrm{supp}(p)} d(v) \leq \sum_{j = 1}^T d_j \leq \frac{1}{\epsilon \alpha}.
    \]
    Finally, to bound the running time, we can implement the algorithm by maintaining a queue of vertices satisfying $\frac{r(u)}{d(u)} \geq \epsilon$ and repeat the push operation on the first item in the queue at every iteration, updating the queue as appropriate.
    The push operation and queue updates can be performed in time proportional to $d(u)$ and so the running time follows from \eqref{eq:pushnum}.
\end{proof}

\paragraph{Analysis of \texttt{LocBipartDC}.}
We will make use of the following fact proved by Andersen et al.\ \cite{ACL2006} which shows that the approximate Pagerank vector will have a large probability mass on a set of vertices if they have low conductance.
\begin{proposition}[\cite{ACL2006}, Theorem~4] \label{lem:apr_escaping_mass}
    For any vertex set $S$ and any constant $\alpha \in [0, 1]$, there is a subset $S_{\alpha} \subseteq S$ with $\vol(S_{\alpha}) \geq \vol(S)/2$ such that for any vertex $v \in S_{\alpha}$, the approximate PageRank vector $\apr(\alpha, \chi_v, r)$ satisfies
    \[
        \apr(\alpha, \chi_v, r)(S) \geq 1 - \frac{\cond{S}}{\alpha} - \vol(S) \max_{u \in V} \frac{r(u)}{\deg(u)}.   
    \]
\end{proposition}
The following proof generalises this result to show that applying the $\sigma$-operator to the approximate Pagerank vector maintains this large probability mass, up to a constant factor.
\begin{proof}[Proof of Lemma~\ref{lem:sapr_escapingmass}]
By Lemma \ref{lem:cond_bip} and Proposition~\ref{lem:apr_escaping_mass}, there is a set $S_\alpha \subseteq S$ with $\vol(S_{\alpha})\geq \vol(S)/2$ such that it holds  for $v \in S_\alpha$ that  
    \[
    	\apr_H(\alpha, \chi_v, r)(L_1 \cup R_2) \geq 1 - \frac{\bipart{L}{R}}{\alpha} - \vol(S) \max_{u \in V} \frac{r(u)}{\deg(u)}.
    \]
    Therefore, by the definition of $p$ we have that\begin{align*}
    	p(L_1 \cup R_2) & \geq \apr(\alpha, \chi_v, r)(L_1 \cup R_2) - \apr(\alpha, \chi_v, r)(L_2 \cup R_1) \\
    	&\geq \apr(\alpha, \chi_v, r)(L_1 \cup R_2) - \apr(\alpha, \chi_v, r)(\overline{L_1 \cup R_2}) \\
    	& \geq \apr(\alpha, \chi_v, r)(L_1 \union R_2) - (1 - \apr(\alpha, \chi_v, r)(L_1 \union R_2)) \\
    	& = 2\cdot\apr(\alpha,\chi_v,r) (L_1\cup R_2) - 1\\
    	& \geq 1   - \frac{2 \bipart{L}{R}}{\alpha} - 2 \vol(S) \max_{u \in V} \frac{r(u)}{\deg(u)},
    \end{align*}
    which proves the statement of the lemma.
\end{proof}
For the proof of Lemma~\ref{lem:sapr_updatestep}, we need to carefully analyse the effect of the $\sigma$-operator on the approximate Pagerank vector, and in particular notice the effect of the residual vector $r$.
\begin{proof}[Proof of Lemma~\ref{lem:sapr_updatestep}]
Let   $u \in V_G$ be an arbitrary vertex. We have that 
    \begin{align*}
       & p(u_1) \\ 
        & = \sigma \circ \big( \pr(\alpha, s) - \pr(\alpha, r)\big)(u_1) \\
        & = \sigma \circ \big(\alpha s + (1 - \alpha)\pr(\alpha, s)W - \alpha r - (1 - \alpha)\pr(\alpha, r)W \big)(u_1) \\
        & = \sigma \circ \big(\alpha s + (1 - \alpha)\apr(\alpha, s, r)W - \alpha r\big)(u_1) \\
        & = \max\Big(0, \big(\alpha s + (1 - \alpha)\apr(\alpha, s, r)W\big)(u_1) - \big(\alpha s + (1 - \alpha)\apr(\alpha, s, r)W\big)(u_2) + \alpha r(u_2) - \alpha r(u_1)\Big) \\
        & \leq \sigma \circ \big(\alpha s + (1 - \alpha)\apr(\alpha, s, r)W\big)(u_1) + \max(0, \alpha r(u_2) - \alpha r(u_1)) \\
        & \leq \sigma \circ \big(\alpha s + (1 - \alpha)\apr(\alpha, s, r)W\big)(u_1) + \alpha r(u_2) \\
        & \leq \sigma \circ \big(\alpha s\big)(u_1) + (1 - \alpha)\sigma \circ \big(\apr(\alpha, s, r) W\big)(u_1) + \alpha r(u_2) \\
        & \leq \alpha s(u_1) + \alpha r(u_2) + (1 - \alpha) (p W) (u_1),
    \end{align*}
    where the first line follows by the definition of the approximate Pagerank, the second line follows by the definition of the Pagerank, and the last two inequalities follow by Lemma~\ref{lem:sigmaproperty}. This proves the first inequality, and the second inequality follows by symmetry.
\end{proof}

We now turn our attention to the Lov\'asz-Simonovits curve defined by the simplified approximate Pagerank vector $p = \sigma \circ \apr(\alpha, s, r)$.
We will bound the curve at some point by the conductance of the corresponding sweep set with the following lemma.
\begin{lemma} \label{lem:ls1}
Let $p = \simplify{\apr(\alpha, s, r)}$. It holds for any $j \in [1, n-1]$ that 
    \begin{align*}
        \lefteqn{\lscurve{\vol(\pjsweep)}}\\
       \leq 
       & \alpha \Big(\lscurve[s]{\vol(\pjsweep)} + \lscurve[r]{\vol(\pjsweep)} \Big)     + (1 - \alpha) \frac{1}{2}\Big( \lscurve{\vol(\pjsweep) - \abs{\partial(\pjsweep)}}\Big)  +   (1 - \alpha)\frac{1}{2}\Big( \lscurve{\vol(\pjsweep) + \abs{\partial(\pjsweep)}}\Big).
    \end{align*}
\end{lemma}
Since $\vol(\pjsweep) \pm \abs{\partial(\pjsweep)} = \vol(\pjsweep)(1 \pm \Phi(\pjsweep))$, this tells us that when the conductance of the sweep set $\pjsweep$ is large, the Lov\'asz-Simonovits curve around the point $\vol(\pjsweep)$ is close to a straight line.
For the proof of Lemma~\ref{lem:ls1}, we will view the undirected graph $H$ instead as a digraph, and each edge $\{u, v\} \in E_H$ as a pair of directed edges $(u, v)$ and $(v, u)$. For any edge $(u, v)$ and vector $p \in \R_{\geq 0}^{2n}$, let
\[
    p(u, v) = \frac{p(u)}{d(u)}
\]
and for a set of directed edges $A$, let
\[
    p(A) = \sum_{(u, v) \in A} p(u, v).
\]
Now for any set of vertices $S \subset V_H$, let $\In(S) = \{(u, v) \in E_H : v \in S\}$ and $\Out(S) = \{(u, v) \in E_H : u \in S\}$.
Additionally, for any simple set $S \subset V_H$, we define the ``simple complement'' of the set $S$ by
\[
    \widehat S = \{u_2 : u_1 \in S\} \union \{u_1 : u_2 \in S\}.
\]
We will need the following property of the Lov\'asz-Simonovitz curve with regard to sets of edges.
\begin{lemma} \label{lem:lsedges}
    For any set of edges $A$, it holds that
   $$
        p(A) \leq \lscurve{\abs{A}}.$$
\end{lemma}
\begin{proof}
    For each $u \in V$, let $x_u = \abs{\{(u, v) \in A\}}$ be the number of edges in $A$ with the tail at $u$.
    Then, we have by definition  that 
    \begin{align*}
        p(A) & = \sum_{(u, v) \in A} \frac{p(u)}{d(u)}  = \sum_{u \in V} \frac{x_u}{d(u)} p(u)
    \end{align*}
    Since $\frac{x_u}{d(u)} \in [0, 1]$ and
    \[
        \sum_{u \in V} \frac{x_u}{d(u)} d(u) = \sum_{u \in V} x_u = \abs{A},
    \]
      the statement follows from equation (\ref{eq:deflscurve}).
\end{proof}
We will also make use of the following lemma from \cite{ACL2006}.
\begin{lemma}[\cite{ACL2006}, Lemma~4]\label{lem:andersenlem4}
    For any distribution $p$, and any set of vertices $A$,
    \[
        pW(A) \leq \frac{1}{2}\cdot 
        \left(p(\In(A) \union \Out(A)) + p(\In(A) \intersect \Out(A))\right)
    \]
\end{lemma}
We now have all the pieces we need to prove the behaviour of the Lov\'asz-Simonovits curve described in Lemma~\ref{lem:ls1}.
\begin{proof}[Proof of Lemma~\ref{lem:ls1}]
By definition, we have that 
\begin{align*}
    \lefteqn{p(S)} \\
    & = \sum_{u_1\in  S} p(u_1) + \sum_{u_2\in S} p (u_2) \\
    & \leq\sum_{u_1\in S}\Big( \alpha \left(s(u_1) + r(u_2)\right) + (1 - \alpha)(p W)(u_1)\Big)  + \sum_{u_2\in S}   \Big(\alpha \left(s(u_2) + r(u_1)\right) + (1 - \alpha)(p W)(u_2) \Big) \\
    & = \alpha\cdot \Bigg(\sum_{u_1\in S} s(u_1)+ \sum_{u_2\in S} s(u_2)  +  \sum_{u_1\in S} r(u_2)+ \sum_{u_2\in S} r(u_1) \Bigg) + (1-\alpha)\cdot\Bigg( \sum_{u_1\in S} (pW)(u_1) + \sum_{u_2\in S} (pW)(u_2) \Bigg)\\
    & = \alpha\cdot \Big( s(S) + r\left(\widehat{S}\right)\!\Big) + (1-\alpha)\cdot pW(S)\\
    & \leq \alpha\cdot \Big( s(S) + r\left(\widehat{S}\right)\!\Big) + (1-\alpha)\cdot \frac{1}{2} \cdot \Big(p(\In(S) \union \Out(S)) + p(\In(S) \intersect\Out(S))\Big),
\end{align*}
where the first inequality follows by Lemma~\ref{lem:sapr_updatestep} and the last inequality follows by  Lemma~\ref{lem:andersenlem4}.
 
 Now, recall that $\lscurve{\vol(\pjsweep)} = p(\pjsweep)$ and notice that $\vol\left(\widehat{\pjsweep}\right) = \vol\left(\pjsweep\right)$ holds by the definition of $\widehat{S}$. Hence, using Lemma~\ref{lem:lsedges}, we have that
 \begin{align*}
        \lefteqn{ \lscurve{\vol(\pjsweep)}}\\
       & =  p(\pjsweep) \\
        & \leq \alpha\left( s(S_j^p) + r\left(\widehat{S_j^p} \right)
         \right)  + (1 - \alpha)\cdot \frac{1}{2}\Big( p\left( \In(\pjsweep) \union \Out(\pjsweep)\right) + p(\In(\pjsweep) \intersect \Out(\pjsweep))\Big) \\
        & \leq  \alpha\Big(\lscurve[s]{\vol(\pjsweep)} + \lscurve[r]{\vol(\pjsweep)}\Big) + (1 - \alpha)\cdot \frac{1}{2}\left( \lscurve{\abs{\In(\pjsweep) \union \Out(\pjsweep)}} + \lscurve{\abs{\In(\pjsweep) \intersect \Out(\pjsweep)}} \right).
    \end{align*}
    It remains to show that
    \[
        \abs{\In(\pjsweep) \intersect \Out(\pjsweep)} = \vol(\pjsweep) - \abs{\partial(\pjsweep)},
    \]
    and
    \[
        \abs{\In(\pjsweep) \union \Out(\pjsweep)} = \vol(\pjsweep) + \abs{\partial(\pjsweep)}.
    \]
    The first follows by noticing that $\abs{\In(\pjsweep) \intersect \Out(\pjsweep)}$ counts precisely twice the number of edges with both endpoints inside $\pjsweep$.
    The second follows since
    \[
        \abs{\In(\pjsweep) \union \Out(\pjsweep)} = \abs{\In(\pjsweep) \intersect \Out(\pjsweep)} + \abs{\In(\pjsweep) \setminus \Out(\pjsweep)} + \abs{\Out(\pjsweep) \setminus \In(\pjsweep)}
    \]
    and both the second and third terms on the right hand side are equal to $\abs{\partial(\pjsweep)}$.
\end{proof}
To understand the following lemma, notice that since $p[0] = 0$ and $p[\vol(H)] = 1$, if the curve were simply linear then it would be the case that $p[k] = \frac{k}{\vol(H)}$ and so the following lemma bounds how far the Lov\'asz-Simonovits curve deviates from this line. In particular, if there is no sweep set with small conductance, then the deviation must be small.
\begin{lemma} \label{lem:lscurve}
     Let $p = \simplify{\apr_H(\alpha, s, r)}$ such that $\max_{u \in V_H} \frac{r(u)}{d(u)} \leq \epsilon$, and  $\phi$   be any constant in $[0, 1]$. If $\cond{\pjsweep} \geq \phi$ for all $j \in [1, \abs{\mathrm{supp}(p)}]$, then
    \[
        p[k] - \frac{k}{\vol(H)} \leq \alpha t + \alpha \epsilon k t + \sqrt{\min(k, \vol(H) - k)}\left(1 - \frac{\phi^2}{8}\right)^t
    \]
    for all $k \in [0, \vol(H)]$ and integer $t \geq 0$.
\end{lemma}
\begin{proof} 
For ease of discussion we write  $k_j = \vol(\pjsweep)$ and let $\overline{k_j} = \min(k_j, \vol(H) - k_j)$. For convenience we will write
    \[
        f_t(k) =   \alpha t + \alpha \epsilon k t + \sqrt{\min(k, \vol(H) - k)}\left(1 - \frac{\phi^2}{8}\right)^t.
    \]
    We will prove by induction that for all $t \geq 0$,
    \[
        p[k] - \frac{k}{\vol(H)} \leq f_t(k).
    \]
    For the base case, this equation holds for $t = 0$ for any $\phi$. Our proof is based on a case distinction:
    (1) For any   integer $k \in [1, \vol(H) - 1]$,
    \[
        p[k] - \frac{k}{\vol(H)}\leq 1 \leq \sqrt{\min(k, \vol(H) - k)} \leq f_0(k);
    \] (2) 
    For $k = 0$ or $k = \vol(H)$, $p[k] - \frac{k}{\vol(H)} = 0 \leq f_0(k)$. The base case follows since $f_0$ is concave, and $p[k]$ is linear between integer values.
    
    Now for the inductive case, we assume that the claim holds for $t$ and will show that it holds for $t+1$.
    It suffices to show that the following holds for every $j \in [1, \abs{\mathrm{supp}(p)}]$:
    \[
        p[k_j] - \frac{k_j}{\vol(H)} \leq f_{t+1}(k_j).
    \]
    The theorem will follow because this equation holds trivially for $j = 0$ and $j = n$ and $p[k]$ is linear between $k_j$ for $j \in \{0, n\} \union [1, |\mathrm{supp}(p)|]$.
    
   For any $j \in [1, |\mathrm{supp}(p)|]$,  it holds  by Lemma~\ref{lem:ls1} that 
   \begin{align*}
        \lscurve{k_j} & \leq \alpha \left(\lscurve[s]{k_j} + \lscurve[r]{k_j} \right) + (1 - \alpha) \frac{1}{2}\left( \lscurve{k_j - \abs{\partial(\pjsweep)}} + \lscurve{k_j + \abs{\partial(\pjsweep)}} \right) \\
        & \leq \alpha + \alpha \epsilon k_j + \frac{1}{2}\left(\lscurve{k_j - \cond{\pjsweep}\cdot\overline{k_j}} + \lscurve{k_j + \cond{\pjsweep}\cdot \overline{k_j}} \right) \\
        & \leq \alpha + \alpha \epsilon k_j + \frac{1}{2}\left( \lscurve{k_j - \phi \overline{k_j}} + \lscurve{k_j + \phi \overline{k_j}}\right),
    \end{align*}
    where the second line uses the fact that $\lscurve[r]{k_j} \leq \epsilon k_j$ and the last line uses the concavity of $\lscurve{k}$ and the fact that $\cond{\pjsweep} \geq \phi$. By the induction hypothesis, we have that
    \begin{align*}
        \lefteqn{\lscurve{k_j}}\\
        & \leq \alpha + \alpha \epsilon k_j + \frac{1}{2}\left( f_t(k_j - \phi \overline{k_j}) + \frac{k_j - \phi \overline{k_j}}{\vol(H)} + f_t(k_j + \phi \overline{k_j}) + \frac{k_j + \phi \overline{k_j}}{\vol(H)} \right) \\
          & = \alpha + \alpha \epsilon k_j + \frac{k_j}{\vol(H)} + \frac{1}{2}\left( f_t(k_j + \phi \overline{k_j}) + f_t(k_j - \phi \overline{k_j}) \right) \\
            &\! \begin{aligned}[t] & =   \alpha + \alpha \epsilon k_j + \alpha t + \alpha \epsilon k_j t +\frac{k_j}{\vol(H)} \\ & \qquad + \frac{1}{2}\left( \sqrt{\min(k_j - \phi \overline{k_j}, \vol(H) - k_j + \phi \overline{k_j})} + \sqrt{\min(k_j + \phi \overline{k_j}, \vol(H) - k_j - \phi \overline{k_j})} \right) \left(1 - \frac{\phi^2}{8}\right)^t \end{aligned} \\
        & \leq   \alpha(t + 1) + \alpha \epsilon k_j (t + 1) + \frac{k_j}{\vol(H)}+ \frac{1}{2}\left(\sqrt{\overline{k_j} - \phi \overline{k_j}} + \sqrt{\overline{k_j} + \phi \overline{k_j}}\right) \left(1 - \frac{\phi^2}{8}\right)^t
    \end{align*}
    Now, using the Taylor series of $\sqrt{1 + \phi}$ at $\phi = 0$, we see the following for any $x \geq 0$ and $\phi \in [0, 1]$.
    \begin{align*}
        \frac{1}{2}\left( \sqrt{x - \phi x} + \sqrt{x + \phi x}\right) & \leq \frac{\sqrt{x}}{2}\left( \left(1 - \frac{\phi}{2} - \frac{\phi^2}{8} - \frac{\phi^3}{16}\right) + \left(1 + \frac{\phi}{2} - \frac{\phi^2}{8} + \frac{\phi^3}{16}\right)  \right) \\
        & \leq \sqrt{x}\left(1 - \frac{\phi^2}{8}\right).
    \end{align*}
    Applying this with $x = \overline{k_j}$, we have
    \begin{align*}
        \lscurve{k_j} - \frac{k_j}{\vol(H)} & \leq   \alpha (t + 1) + \alpha \epsilon k_j (t+1) + \sqrt{\overline{k_j}}\left(1 - \frac{\phi^2}{8}\right) \left(1 - \frac{\phi^2}{8}\right)^t \\
        & = f_{t+1}(k_j),
    \end{align*}
    which completes the proof.
\end{proof}
The next lemma follows almost directly from Lemma~\ref{lem:lscurve} by carefully choosing the value of $t$.
\begin{proof}[Proof of Lemma~\ref{lem:probimpliescond}]
Let $H_c$ be a weighted version of the double cover with the weight of every edge equal to some constant $c$.
    Let $k = \vol_{H_c}(S)$ and $\overline{k} = \min(k, \vol(H_c) - k)$.
    Notice that the behaviour of random walks and the conductance of vertex sets are unchanged by re-weighting the graph by a constant factor.
    As such, $p = \sigma \circ \apr_{H_c}(\alpha, s, r)$ and we let $\phi = \min_j \cond[H_c]{\pjsweep}$. Hence, by Lemma~\ref{lem:lscurve} it holds for any $t>0$ that
    \begin{align*}
        \delta & < \lscurve{k} - \frac{k}{\vol(H_c)}  < \alpha t (1 + \epsilon k) + \sqrt{\overline{k}} \left(1 - \frac{\phi^2}{8}\right)^t.
    \end{align*}
    Letting $v = \vol(H_c)$ we set
    $t = \left\lceil \frac{8}{\phi^2} \ln \frac{2 \sqrt{v/2}}{\delta} \right\rceil$.
    Now, assuming that
    \begin{equation} \label{eq:deltabound}
        \delta \geq \frac{2}{\sqrt{v / 2}}
    \end{equation}
    we have
\[
        \delta < (1 + \epsilon k) \alpha \left\lceil \frac{8}{\phi^2} \ln \frac{2 \sqrt{v/2}}{\delta}\right\rceil + \frac{\delta}{2 \sqrt{v/2}} \sqrt{\overline{k}} 
        < \frac{9 (1 + \epsilon k) \alpha \ln(v/2)}{\phi^2} + \frac{\delta}{2}
        \]
    We set $c$ such that
    \[
        v = c\ \vol(H) = \vol(H_c) = \left(\frac{4}{\delta}\right)^2
    \]
    which satisfies \eqref{eq:deltabound} and implies that \[
        \phi< \sqrt{\frac{18 (1 + \epsilon k) \alpha \ln(4 / \delta)}{\delta}}.
        \]
By the definition of $\phi$, there must be some $S_j^p$ with the desired property.
\end{proof}
Finally, we prove the performance guarantees of Algorithm~\ref{alg:local_max_cut} in the following theorem.
\begin{proof}[Proof of Theorem~\ref{thm:main_thm}]
We set the parameters of the main algorithm to be
\begin{itemize}
    \item $\alpha = \widehat{\beta}^2/378 = 20\beta$ 
    \item $\epsilon=1/(20\gamma)$. 
\end{itemize} \ 
Let  $p = \simplify{\apr_H(\alpha, \chi_v, r)}$ be the simplified approximate Pagerank vector computed in the algorithm, which satisfies that $\max_{v \in V_H} \frac{r(v)}{d(v)} \leq \epsilon$, and let $C'=L_1\cup R_2$ be the target set in the double cover. Hence, by  Lemma~\ref{lem:sapr_escapingmass} we have that
\[
    p(C')  \geq 1 - \frac{2 \beta}{\alpha} - 2 \epsilon\cdot  \vol(C)  \geq 1 - \frac{1}{10} - \frac{1}{10} = \frac{4}{5}.
\]
Notice that, since $C'$ is simple, we have that $\frac{\vol(C')}{\vol(V_H)} \leq \frac{1}{2}$ and
\[
    p(C') - \frac{\vol(C')}{\vol(V_H)} \geq \frac{4}{5} - \frac{1}{2} = \frac{3}{10}.
\]
Therefore, by Lemma~\ref{lem:probimpliescond} and choosing $\delta=3/10$, there is some $j \in [1, \abs{\mathrm{supp}(p)}]$ such that
\begin{align*}
    \cond{\pjsweep} & < \sqrt{\frac{360}{3}(1 + \epsilon \vol(C))\alpha \ln(\frac{40}{3})} \leq \sqrt{\left(1 + \frac{1}{20}\right) 7200 \beta}   = \widehat \beta.
\end{align*}
Since such a $S_j^p$ is a simple set, by Lemma~\ref{lem:cond_bip}, we know that the sets $L' = \{u \in V_G : u_1 \in \pjsweep \}$ and $R' = \{u \in V_G : u_2 \in \pjsweep \}$ satisfy $\bipart(L', R') \leq \widehat \beta$.
    For the volume guarantee, Lemma~\ref{lem:aprdc} gives that
    \[
        \vol(\mathrm{supp}(p)) \leq \frac{1}{\epsilon \alpha} = \frac{\gamma}{\beta}.
    \]
    Since the sweep set procedure in the algorithm is over the support of $p$, we have that
    $$
        \vol(L' \union R') \leq \vol(\mathrm{supp}(p)) \leq \beta^{-1} \gamma.$$
    Now, we will analyse the running time of the algorithm.
    By Lemma~\ref{lem:aprdc}, computing the approximate Pagerank vector $p'$ takes time
    \[
        \bigo{\frac{1}{\epsilon \alpha}} = \bigo{\beta^{-1}\gamma}.
    \]
    Since $\vol(\mathrm{supp}(p')) = \bigo{\frac{1}{\epsilon \alpha}}$, computing $\simplify{p'}$ can also be done in time $\bigo{\beta^{-1}\gamma}$.
    The cost of checking each conductance in the sweep set loop is $\bigo{\vol(\mathrm{supp}(p)) \ln(n)} = \bigo{\beta^{-1} \gamma \ln(n)}$. This dominates the running time of the algorithm, which completes the proof of the theorem.
\end{proof}

\section{Omitted detail from Section~\ref{sec:directed}\label{sec:directedap}}
 In this section, we present all the omitted details from Section~\ref{sec:directed} of our paper. We will use the evolving set process to analyse the correctness and performance of the \texttt{EvoCutDirected} algorithm.
\paragraph{The evolving set process.}
     To understand the ESP update process, notice that $\chi_v W \chi_{S_i}^\transpose$ is the probability that a one-step random walk starting at $v$ ends inside $S_i$.
    Given that the transition kernel of the evolving set process is given by $\mathbf{K}(S, S') = \p[S_{i+1} = S' | S_i = S]$, the \emph{volume-biased} ESP is defined by the transition kernel
    \[
        \mathbf{\widehat K}(S, S') = \frac{\vol(S')}{\vol(S)} \mathbf{K}(S, S')
    \]
We now give some useful facts about the evolving set process which we will use in our analysis of Algorithm~\ref{alg:directed_evo_cut}.
The first proposition tells us that an evolving set process running for sufficiently many steps is very likely to have at least one state with low conductance.

    \begin{proposition}[\cite{AP2009}, Corollary 1] \label{prop:esp_cond}
        Let $S_0, S_1, \ldots, S_T$ be sets drawn from the volume-biased evolving set process beginning at $S_0$. Then, it holds  for any constant $c \geq 0$ that 
        \[
            \p\left[\min_{j \leq T} \Phi(S_j) \leq \sqrt{c} \sqrt{4 T^{-1} \ln (\vol(V))}\right] \geq 1 - \frac{1}{c}.
        \]
    \end{proposition}
        If $(X_0, X_1, \ldots, X_T)$ is a lazy random walk Markov chain with $X_0 = x$ then for any set $A \subseteq V_H$ and integer $T$, let
    \[
        \mathrm{esc}(x, T, A) \triangleq \p \left[\bigcup_{j = 0}^T X_j \not \in A\right]
    \]
    be the probability that the random walk leaves the set $A$ in the first $T$ steps.
    We will use the following results connecting the escape probability to the conductance of the set $A$ and the volume of the sets in the evolving set process.
    \begin{proposition}[\cite{ST2013}, Proposition 2.5] \label{prop:escape_prob}
        There is a set $A_T \subseteq A$ with $\vol(A_T) \geq \frac{1}{2}\vol(A)$, such that it holds for any $x \in A_T$ that
        $$
            \mathrm{esc}(x, T, A) \leq T \Phi(A).$$
    \end{proposition}
    \begin{proposition}[\cite{AP2009}, Lemma 2] \label{prop:esp_vol}
        For any vertex $x$, let $S_0, S_1, \ldots, S_T$ be sampled from a volume-biased ESP starting from $S_0 = \{x\}$.
        Then, it holds for any set $A \subseteq V_H$ and $\lambda > 0$ that 
        \[
            \p\left[\max_{t \leq T} \frac{\vol(S_t \setminus A)}{\vol(S_t)} > \lambda \mathrm{esc}(x, T, A)\right] < \frac{1}{\lambda}.
        \]
    \end{proposition}

\paragraph{Analysis of \texttt{EvoCutDirected}.}
We now come to the formal analysis of our algorithm for finding almost-bipartite sets in digraphs.
We start by proving the connection between $F(L, R)$ in a digraph with the conductance of the set $L_1 \union R_2$ in the semi-double cover.
\begin{proof}[Proof of Lemma~\ref{lem:flow_cond}]
We define  $S_H = L_1 \union R_2$, and have that  
\begin{align*}
\Phi_H(S_H) & = \frac{e(S_H, V_H\setminus S_H)}{\vol_H(S_H) }\\
& = \frac{\vol_H(S_H) - 2e_H(S_H, S_H)}{\vol_{H}(S_H) } \\
& = 1  - \frac{2e_H(S_H, S_H)}{\vol_{H}(S_H)} \\
& =   1 - \frac{2 e_G(L, R)}{\vol_{\mathrm{out}}(L) + \vol_{\mathrm{in}}(R)} \\
& = F_G(L,R).
\end{align*}
This proves the first statement of the lemma.  
    The second statement of the lemma follows by the same argument.
\end{proof}
 We now show how removing all vertices $u$ where $u_1 \in S$ and $u_2 \in S$ affects the conductance of $S$ on the semi-double cover.
 \begin{lemma} \label{lem:lazysimplify}
For  any $\epsilon$-simple set $S\subset V_H$,
let $P = \{u_1, u_2 : u \in V_G, u_1 \in S \mbox{ and } u_2 \in S\}$
and set $S'=S\setminus P$. 
Then,  
 $
        \cond{S'} \leq \frac{1}{1 - \epsilon} (\cond{S} + \epsilon).
    $
\end{lemma}
\begin{proof} 
By the definition of conductance, we have
    \begin{align*}
        \cond{S'} & = \frac{e(S', V \setminus S')}{\vol(S')} \\
        & \leq \frac{e(S, V \setminus S) + \vol(P)}{\vol(S) - \vol(P)} \\
        & \leq \frac{1}{1 - \epsilon} \frac{e(S, V \setminus S) + \vol(P)}{\vol(S)} \\
        & \leq \frac{1}{1 - \epsilon} \left( \frac{e(S, V \setminus S)}{\vol(S)} + \epsilon \right) \\
        & = \frac{1}{1 - \epsilon} \left( \cond{S} + \epsilon \right),
    \end{align*}
    which proves the lemma.
\end{proof}
As such, we would like to show that the output of the evolving set process is very close to being a simple set.
The following lemma shows that, by choosing the number of steps of the evolving set process carefully, we can very tightly control the overlap of the output of the evolving set process with the target set $S$, which we know to be simple.
    \begin{lemma} \label{lem:esp_tightvol}
        Let $C \subset V_H$ be a set with conductance $\Phi(C) = \phi$, and let  $T = (100 \phi^{\frac{2}{3}})^{-1}$. There exists a set $C_g \subseteq C$ with volume at least $\vol(C) / 2$ such that for any $x \in C_g$, with probability at least $7 / 9$, a sample path $(S_0, S_1, \ldots, S_T)$ from the volume-biased ESP started from $S_0 = \{x\}$ will satisfy the following:
        \begin{itemize}
            \item $\Phi(S_t) < 60 \phi^{\frac{1}{3}} \ln^{\frac{1}{2}}\vol(V)$ for some $t \in [0, T]$;
            \item $\vol(S_t \intersect C) \geq \left(1 - \frac{1}{10}\phi^{\frac{1}{3}}\right) \vol(S_t)$ for all $t \in [0, T]$. 
        \end{itemize}
    \end{lemma}
\begin{proof}
    By Proposition~\ref{prop:esp_cond}, with probability at least $(1 - 1/9)$ there is some $t < T$ such that
    \begin{align*}
        \Phi(S_t) & \leq 3 \sqrt{4 T^{-1} \ln \vol(V_H)}  = 3 \sqrt{400 \phi^{\frac{2}{3}} \ln \vol(V_H)}  = 60 \phi^{\frac{1}{3}} \ln^{\frac{1}{2}} (\vol(V_H)).
    \end{align*}
    Then, by Proposition~\ref{prop:escape_prob} we have that 
    $
        \mathrm{esc}(x, T, C) \leq T \phi = \frac{1}{100} \phi^{1/3}$.
     Hence, by Proposition~\ref{prop:esp_vol}, with probability at least $9/10$,
    \[
    \max_{t\leq T} \frac{\vol(S_t\setminus C)}{\vol(S_t)} \leq 10\cdot \mathrm{esc}(x,T,C) \leq \frac{1}{10}\cdot \phi^{1/3}
    \]
    Applying the union bound on these two events proves the statement of the lemma. 
\end{proof}
Now we combine everything together to prove Theorem~\ref{thm:directedresult}.
 
\begin{proof}[Proof of Theorem~\ref{thm:directedresult}]
By Lemma~\ref{lem:esp_tightvol}, we know that, with  probability at least $7/9$,  the set $S$ computed by the algorithm will satisfy
    \[
        \vol(S \intersect C) \geq \left(1 - \frac{1}{10}\phi^{\frac{1}{3}}\right) \vol(S).
    \]
    Let $P = \{u_1, u_2 : u_1 \in S $ and $u_2 \in S\}$. Then, since $C$ is simple, we have that
    \begin{align*}
        \vol(P) & \leq 2\cdot (\vol(S) - \vol(S \intersect C))   \leq 2\cdot \left(1 - \left(1 - \frac{1}{10} \phi^{\frac{1}{3}}\right)\right)\cdot \vol(S)  = \frac{1}{5}\cdot  \phi^{\frac{1}{3}} \vol(S),
    \end{align*}
    which implies that $S$ 
    is $(1/5)\cdot \phi^{\frac{1}{3}}$-simple. 
    Therefore, by letting $S' = S \setminus P$, we have by Lemma~\ref{lem:lazysimplify} that 
    \begin{align*}
        \Phi(S') & \leq \frac{1}{1 - \epsilon}\cdot  \left(\Phi(S) + \epsilon\right) 
          \leq \frac{1}{1 - \frac{1}{5}} \left(\Phi(S) + \frac{1}{5}\phi^{\frac{1}{3}}\right)   \leq \frac{5}{4}\cdot  \Phi(S) + \frac{1}{4}\cdot  \phi^{\frac{1}{3}}.  
    \end{align*}
    Since $\Phi(S) = \bigo{\phi^{\frac{1}{3}} \ln^{\frac{1}{2}}(n)}$ by Lemma~\ref{lem:esp_tightvol}, we   have that $\Phi(S') = \bigo{\phi^{\frac{1}{3}} \ln^{\frac{1}{2}}(n)}$.
    For the volume guarantee, direct calculation shows that 
    \begin{align*}
        \vol(S') \leq \vol(S)  \leq \frac{\vol(S \intersect C)}{1 - \frac{1}{10}\phi^{\frac{1}{3}}} \leq \frac{\vol(C)}{1 - \phi^{\frac{1}{3}}}.
    \end{align*}
    Finally, the running time of the algorithm is dominated by the $\texttt{GenerateSample}$ method, which was shown in \cite{AP2009} to have running time $\bigo{\vol(C) \phi^{-\frac{1}{2}} \ln^{\frac{3}{2}}(n)}$.
\end{proof}

\end{document}

%% file: basic.tikzstyles
\tikzstyle{basic}=[fill=white, draw=black, shape=circle]
\tikzstyle{big dashed}=[fill=white, draw=black, shape=circle, minimum width=1cm, dashed]
\tikzstyle{vertical ellipse dashed}=[fill=white, draw=black, minimum width=0.75cm, minimum height=3cm, ellipse, dashed, tikzit shape=rectangle, tikzit draw=magenta]
\tikzstyle{red}=[fill=red, draw=black, shape=circle]
\tikzstyle{green}=[fill={rgb,255: red,0; green,128; blue,128}, draw=black, shape=circle]
\tikzstyle{blue}=[fill=blue, draw=black, shape=circle]
\tikzstyle{huge dashed}=[fill=white, draw=black, shape=circle, dashed, minimum width=2cm]
\tikzstyle{medium}=[fill=white, draw=black, shape=circle, minimum width=1cm]
\tikzstyle{pale green}=[fill={rgb,255: red,173; green,231; blue,0}, draw=black, shape=circle, minimum width=1cm]
\tikzstyle{horizontal ellipse dashed}=[fill=white, draw=black, tikzit draw=magenta, tikzit shape=rectangle, minimum width=3cm, minimum height=0.75cm, ellipse, dashed]
\tikzstyle{minsize}=[fill=white, draw=black, shape=circle, minimum width=0.75cm]
\tikzstyle{horizontal ellipse green}=[fill={rgb,255: red,191; green,255; blue,0}, draw=black, tikzit draw={rgb,255: red,191; green,255; blue,0}, tikzit shape=rectangle, minimum width=3cm, minimum height=0.75cm, ellipse, dashed]
\tikzstyle{horizontal ellipse blue}=[fill={rgb,255: red,107; green,203; blue,255}, draw=black, tikzit draw=blue, tikzit shape=rectangle, minimum width=3cm, minimum height=0.75cm, ellipse, dashed]

\tikzstyle{directed}=[->, line width=1pt]
\tikzstyle{undirected}=[-, line width=1pt]
\tikzstyle{directed red}=[draw=red, ->, line width=1pt]
\tikzstyle{directed green}=[draw={rgb,255: red,0; green,128; blue,128}, ->, line width=1pt]
\tikzstyle{directed blue}=[draw=blue, ->, line width=1pt]
\tikzstyle{directed purple}=[draw={rgb,255: red,128; green,0; blue,128}, ->, line width=1pt]
\tikzstyle{undirected red}=[-, draw=red, line width=1pt]
\tikzstyle{undirected green}=[-, draw={rgb,255: red,0; green,128; blue,128}, line width=1pt]
\tikzstyle{undirected blue}=[-, draw=blue, line width=1pt]
\tikzstyle{undirected purple}=[-, draw={rgb,255: red,128; green,0; blue,128}, line width=1pt]
\tikzstyle{undirected dashed}=[-, line width=1pt, dashed]
\tikzstyle{orange dashed}=[-, draw={rgb,255: red,255; green,128; blue,0}, dashed, line width=1.5pt]

%% file: commands.tex
\newtheorem{theorem}{Theorem}

\newtheorem{lemma}{Lemma}

\newtheorem{definition}{Definition}

\newtheorem{proposition}{Proposition}

\newcommand{\transpose}{\intercal}                    
\newcommand{\vol}{\mathrm{vol}}                 
\newcommand{\sigmap}{(\sigma\circ p)}

\newcommand{\R}{\mathbb{R}}

\definecolor{indiagreen}{rgb}{0.07, 0.53, 0.03}


%% file: figures/dcconstruct2.tikz
\begin{tikzpicture}
	\begin{pgfonlayer}{nodelayer}
		\node [style=basic] (0) at (0, 1.5) {$a_1$};
		\node [style=basic] (1) at (3, 1.5) {$b_1$};
		\node [style=basic] (2) at (0, -1.5) {$a_2$};
		\node [style=basic] (3) at (3, -1.5) {$b_2$};
		\node [style=basic] (4) at (6, 1.5) {$c_1$};
		\node [style=basic] (5) at (9, 1.5) {$d_1$};
		\node [style=basic] (6) at (9, -1.5) {$d_2$};
		\node [style=basic] (7) at (6, -1.5) {$c_2$};
		\node [style=basic] (8) at (-9, 1.5) {a};
		\node [style=basic] (9) at (-6, 1.5) {b};
		\node [style=basic] (10) at (-9, -1.5) {c};
		\node [style=basic] (11) at (-6, -1.5) {d};
		\node [style=none] (12) at (-3, 0) {$\Longrightarrow$};
	\end{pgfonlayer}
	\begin{pgfonlayer}{edgelayer}
		\draw [style=undirected] (8) to (9);
		\draw [style=undirected blue] (10) to (11);
		\draw [style=undirected red] (8) to (11);
		\draw [style=undirected green] (8) to (10);
		\draw [style=undirected] (0) to (3);
		\draw [style=undirected] (2) to (1);
		\draw [style=undirected red] (0) to (6);
		\draw [style=undirected red] (2) to (5);
		\draw [style=undirected green] (0) to (7);
		\draw [style=undirected green] (4) to (2);
		\draw [style=undirected blue] (4) to (6);
		\draw [style=undirected blue] (5) to (7);
	\end{pgfonlayer}
\end{tikzpicture}

%% file: figures/nonsimple2.tikz
\begin{tikzpicture}
	\begin{pgfonlayer}{nodelayer}
		\node [style=none] (11) at (-4.5, 2) {\huge{$\cdots$}};
		\node [style=basic] (21) at (-1.5, 2) {$a_1$};
		\node [style=none] (1) at (-3.2, 3) {$S$};
		\node [style=basic] (31) at (1.5, 2) {$b_1$};
		\node [style=basic] (41) at (4.5, 2) {$c_1$};
		\node [style=basic] (51) at (7.5, 2) {$d_1$};
		\node [style=basic] (61) at (10.5, 2) {$e_1$};
		\node [style=basic] (71) at (13.5, 2) {$f_1$};
		\node [style=none] (81) at (16.5, 2) {\huge{$\cdots$}};
		\node [style=none] (12) at (-4.5, -2) {\huge{$\cdots$}};
		\node [style=basic] (22) at (-1.5, -2) {$a_2$};
		\node [style=basic] (32) at (1.5, -2) {$b_2$};
		\node [style=basic] (42) at (4.5, -2) {$c_2$};
		\node [style=basic] (52) at (7.5, -2) {$d_2$};
		\node [style=basic] (62) at (10.5, -2) {$e_2$};
		\node [style=basic] (72) at (13.5, -2) {$f_2$};
		\node [style=none] (82) at (16.5, -2) {\huge{$\cdots$}};
	\end{pgfonlayer}
	\begin{pgfonlayer}{edgelayer}
	    \draw[fill=blue!20] plot [smooth cycle, tension=0.4] coordinates {(-2.3,2.9) (14.3,2.9) (11,-3) (1,-3)};
		\draw [style=dashed] (71) to (82.center);
		\draw [style=dashed] (81.center) to (72);
		\draw [style=dashed] (21) to (12.center);
		\draw [style=dashed] (11.center) to (22);
		\draw [style=undirected] (31) to (42);
		\draw [style=undirected] (32) to (41);
		\draw [style=undirected] (41) to (52);
		\draw [style=undirected] (42) to (51);
		\draw [style=undirected] (31) to (52);
		\draw [style=undirected] (32) to (51);
		\draw [style=undirected] (51) to (62);
		\draw [style=undirected] (52) to (61);
		\draw [style=undirected] (61) to (72);
		\draw [style=undirected] (62) to (71);
		\draw [style=undirected] (21) to (32);
		\draw [style=undirected] (22) to (31);
		\draw [style=undirected] (41) to (62);
		\draw [style=undirected] (42) to (61);
	\end{pgfonlayer}
\end{tikzpicture}

%% file: figures/dcconstruct.tikz
\begin{tikzpicture}
	\begin{pgfonlayer}{nodelayer}
		\node [style=basic] (0) at (0, 1.5) {$a_1$};
		\node [style=basic] (1) at (3, 1.5) {$b_1$};
		\node [style=basic] (2) at (0, -1.5) {$a_2$};
		\node [style=basic] (3) at (3, -1.5) {$b_2$};
		\node [style=basic] (4) at (6, 1.5) {$c_1$};
		\node [style=basic] (5) at (9, 1.5) {$d_1$};
		\node [style=basic] (6) at (9, -1.5) {$d_2$};
		\node [style=basic] (7) at (6, -1.5) {$c_2$};
		\node [style=basic] (8) at (-9, 1.5) {a};
		\node [style=basic] (9) at (-6, 1.5) {b};
		\node [style=basic] (10) at (-9, -1.5) {c};
		\node [style=basic] (11) at (-6, -1.5) {d};
		\node [style=none] (12) at (-3, 0) {$\Longrightarrow$};
	\end{pgfonlayer}
	\begin{pgfonlayer}{edgelayer}
		\draw [style=undirected blue] (4) to (2);
		\draw [style=undirected] (0) to (3);
		\draw [style=undirected purple] (1) to (7);
		\draw [style=undirected red] (1) to (6);
		\draw [style=undirected green] (4) to (6);
		\draw [style=directed] (8) to (9);
		\draw [style=directed purple] (9) to (10);
		\draw [style=directed red] (9) to (11);
		\draw [style=directed blue] (10) to (8);
		\draw [style=directed green] (10) to (11);
	\end{pgfonlayer}
\end{tikzpicture}

%% file: figures/counterexample.tikz
\begin{tikzpicture}
	\begin{pgfonlayer}{nodelayer}
		\node [style=basic] (0) at (0, 1.5) {$a_1$};
		\node [style=basic] (1) at (3, 1.5) {$b_1$};
		\node [style=basic] (2) at (0, -1.5) {$a_2$};
		\node [style=basic] (3) at (3, -1.5) {$b_2$};
		\node [style=basic] (8) at (-9, 0) {a};
		\node [style=basic] (9) at (-6, 0) {b};
		\node [style=none] (12) at (-3, 0) {$\Longrightarrow$};
	\end{pgfonlayer}
	\begin{pgfonlayer}{edgelayer}
		\draw [style=undirected blue] (0) to (3);
		\draw [style=directed blue] (8) to (9);
	\end{pgfonlayer}
\end{tikzpicture}